\documentclass[aps,nobibnotes,twocolumn,superscriptaddress,showkeys,thelongbibliography]{revtex4-2}

\usepackage{dcolumn}

\usepackage{etoolbox}

\usepackage{graphicx}
\usepackage{enumerate}
\usepackage{amsmath}
\usepackage{amssymb}  
\usepackage{mathtools}
\usepackage{bm}  
\usepackage{pdfsync}  
\usepackage{hyperref}
\usepackage{physics}
\usepackage{natbib}
\usepackage{comment}
\usepackage{float}

\usepackage{amsthm}
\newtheorem{theorem}{Theorem}[]
\newtheorem{corollary}[theorem]{Corollary}
\newtheorem{definition}{Definition}

\begin{document}

\title{Generating pseudo-random unitaries with a Floquet driven chaotic quantum system}
\author{Alice C. Quillen}
\email{aquillen@ur.rochester.edu}
\affiliation{Department of Physics and Astronomy, University of Rochester,  Rochester, NY, 14627, USA}
\author{Abobakar Sediq Miakhel}  \email{amiakhel@u.rochester.edu}
\affiliation{Department of Physics and Astronomy, University of Rochester,  Rochester, NY, 14627, USA}

\begin{abstract}
We explore using an ergodic Floquet quantum system on a torus to generate pseudo-random unitary operators.   We choose a regime of the perturbed Harper model  with strong perturbations and perturbation frequency exceeding the libration frequency to ensure that the system has an ergodic region that covers phase space and lacks resonant substructure. We generate a sample of unitary operators in a finite dimensional space by computing Floquet propagators from a distribution of its control parameters. To compare the distribution of unitaries to that of a Haar-random distribution, we compute k-frame potentials from samples of numerically generated unitaries.  We find that uniform distributions of 4 control parameters can generate an approximate 3-design.   Distributions of fewer control parameters are required to create an approximate 3-design if the Floquet system parameters drift.  
\end{abstract}

\maketitle
\tableofcontents


\section{Introduction}

A mechanism or recipe for generating a set of random unitary operators, known as a quantum sampler,  or a generator of pseudo-random unitary operators,  is considered a resource in quantum information science. 
Quantum random sampling has been used in benchmarking computational quantum advantage \citep{Bremner_2017,Boixo_2018,Arute_2019,Hangleiter_2023},  quantum cryptography \citep{Prabhanjan_2022} and quantum algorithms \citep{Sen_2006}. Random quantum sampling can be applied to generate random numbers \citep{Ma_2016,Brask_2017}, mimic the process of thermalization \citep{Temme_2011,Goold_2016},  rapidly delocalize, entangle or scramble information \citep{Brown_2012,Susskind_2016,Nahum_2017}, or find an approximate classical description of a quantum state with only a few measurements \citep{Huang_2020}. 

Quantum random samplers include random quantum circuits which are families of circuits composed of sequences of quantum gates drawn independently and randomly from a set of unitary operators consisting of quantum gates operating on small subsets of qubits, typically containing only one or two qubits \citep{Emerson_2003,Emerson_2005,Harrow_2009,Boixo_2018,Harrow_2023,Fisher_2023}. 
Generalized Fibonacci drives, where operators are drawn from a universal gate set and applied in an order that is generated from an aperiodic sequence, are another approach to generating a  distribution 
of unitaries \citep{PilatowskyCameo_2023}.  Alternatively, random unitaries can be generated by projecting an ensemble derived from a chaotic many body system \citep{Cotler_2023}. 

For a quantum system of dimension $N$, 
a distribution of unitary operators generated from a random quantum circuit should be difficult to distinguish from a set of unitary operators that is uniformly distributed according to the Haar measure
on the unitary group U$(N)$ of dimension $N$ \citep{Ji_2018}.  
The notion of a k- or t-design can be used to 
 quantitatively describe how well a set of 
pseudo-random operators behaves like a distribution which is uniform respect to the Haar measure \citep{Gross_2007,Dankert_2009,Harrow_2009,Brandao_2016}. 
A k-design is a set of unitary operators
in an $N$ dimensional Hilbert space that has k-th moments equal to those of the uniform distribution with respect to the Haar measure (e.g., \citep{Low_2010,Webb_2016,Zhu_2017,Mele_2024}). 
A distribution of unitary operators that has t-th moments near those of a Haar-uniform distribution 
is called an approximate t-design \citep{Brandao_2016}. 
In an $N$ dimensional system, a sequence of unitaries that are diagonal in one basis, alternating with
unitaries that are diagonal in that of the Fourier transformed basis, gives an approximate 2-design \citep{Nakata_2017}.  Sufficiently deep but 
local random quantum circuits on systems of qubits are approximate t-designs \citep{Brandao_2016}. 

Time-periodic or Floquet driving is a way to control the properties and dynamics of quantum systems \citep{Rudner_2020} and for carrying out quantum computations \citep{Farhi_2000,Albash_2018}.  Periodically perturbed or periodically kicked systems are examples of classical and related quantum chaotic systems \citep{Chirikov_1979,Casati_1979,Chirikov_1988,Izrailev_1989,Wei_1991,Quillen_2025}. 
Chaos, complexity and pseudo-randomness are connected via how perturbed operators can 
diverge from each other as a function of time \citep{Scott_2006,Brown_2012,Roberts_2017,Brown_2018}.  
Motivated by the possible connection between chaos and pseudo-randomness, 
we explore how a chaotic system on the torus could aid in generating pseudo-random unitaries. 

\begin{figure*}[htbp]\centering 
\includegraphics[height=2truein]{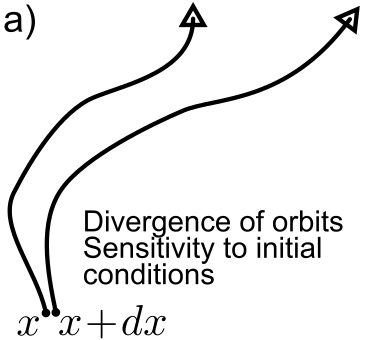}
\includegraphics[height=2truein]{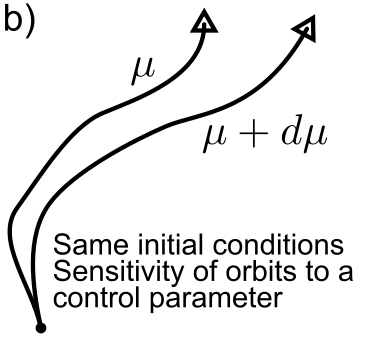}
\includegraphics[height=2truein]{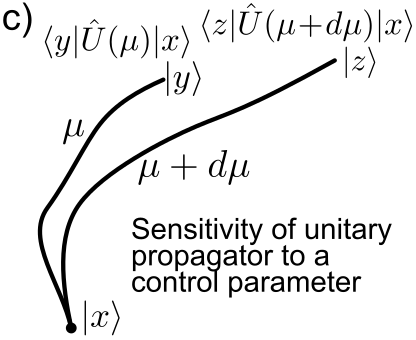}
\caption{a) We show an illustration of how orbits in a chaotic system are sensitive to initial conditions. b)  If the Lyapunov exponent, measuring the divergence rate of orbits, is sensitive to the parameter $\mu$, then orbits begun at the same initial conditions but with slightly different parameters would diverge. c) 
A quantized version of the classical system shown in b) would have transition probabilities between localized states that are sensitive to the parameter $\mu$.  The unitary propagator operator would be sensitive to the parameter $\mu$. 
\label{fig:illust}}
\end{figure*}

\begin{figure}[htbp]\centering
\includegraphics[width=2.5truein]{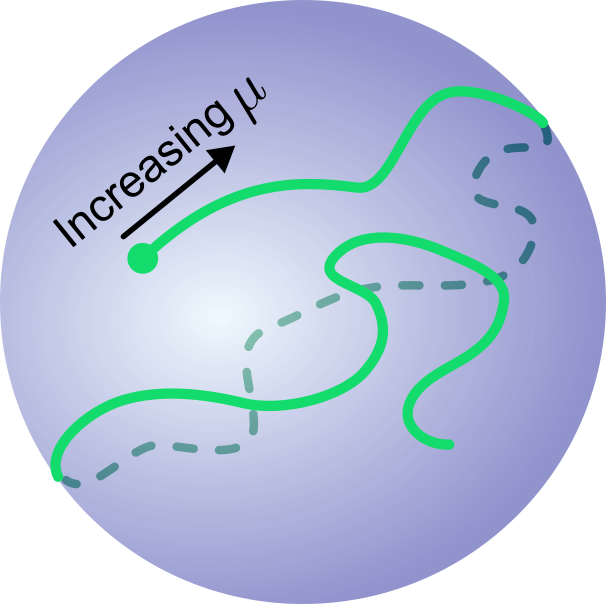}
\caption{The surface of the blue sphere represents the space of unitary operators 
in $N$ dimensions; U$(N)$.   The green trajectory
represents how a unitary operator (or propagator) $\hat U(\mu)$ depends on 
a parameter $\mu$.  
 As $\mu$ increases, the propagator wanders.  
 If the propagator is highly sensitive to 
 the parameter $\mu$, then a range of $\mu$ can give unitaries that are widely distributed 
 across U$(N)$.  \label{fig:UN} }
\end{figure}

\begin{figure*}[!ht]\centering 
\includegraphics[width=5truein, trim = 0 10 0 0,clip]{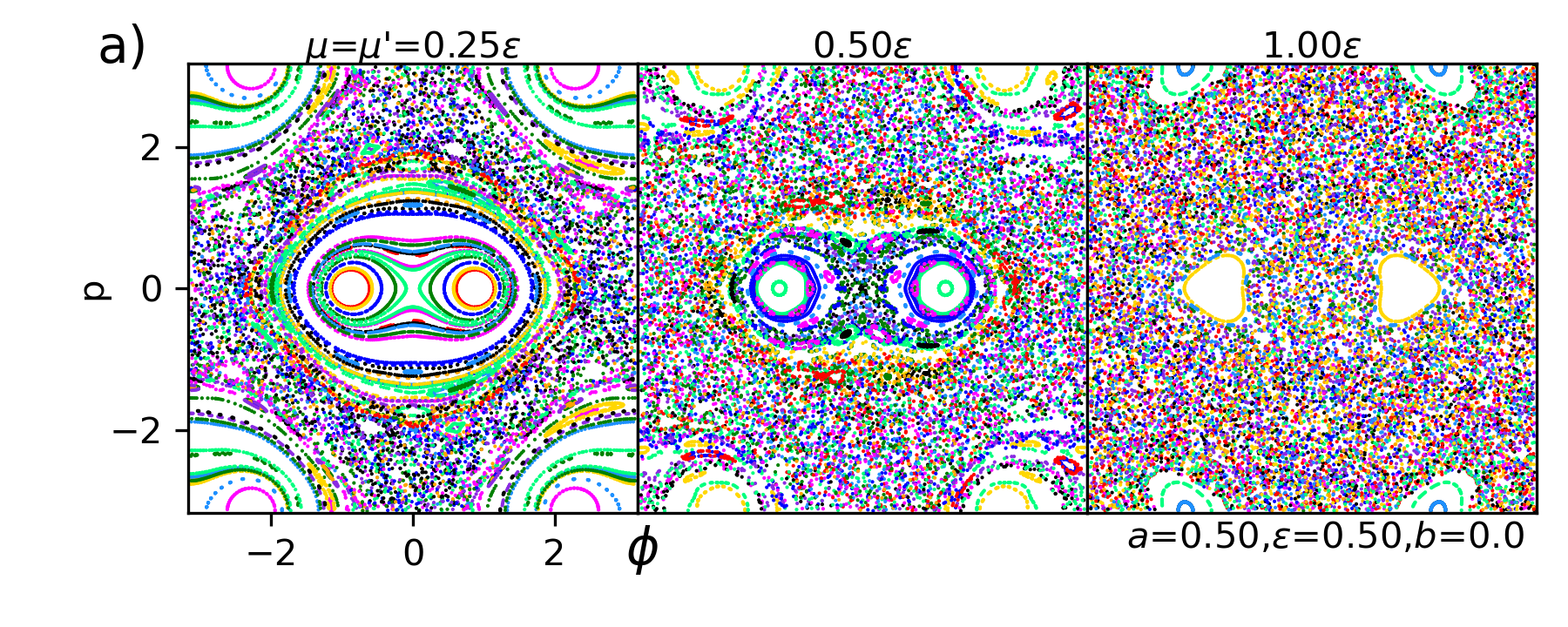}
\includegraphics[width=5truein, trim = 0 10 0 0,clip]{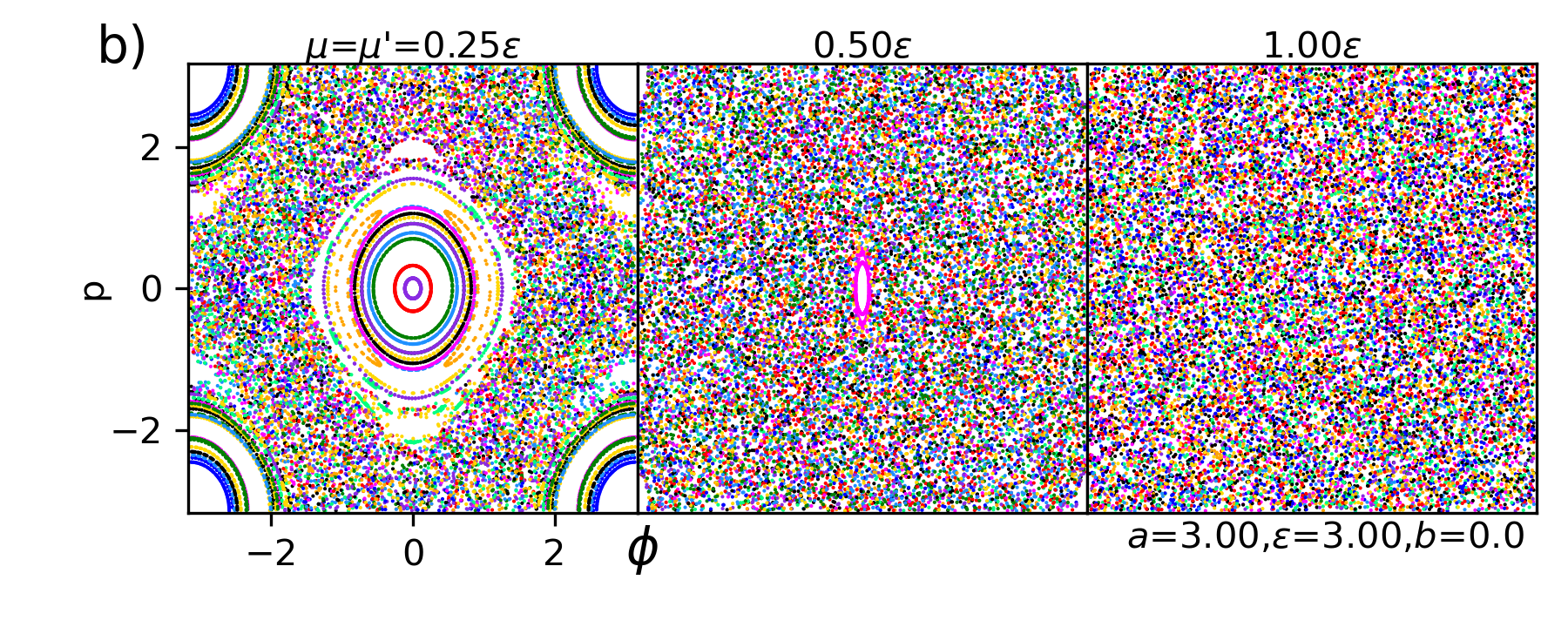}
\caption{a) We show Poincar\'e maps for the classical Hamiltonian system 
given by equations \ref{eqn:Hclass}, \ref{eqn:H0}, \ref{eqn:H1} with ratio of perturbation frequency to 
libration frequency 
$\lambda = (a \epsilon)^{-\frac{1}{2}} = 2$. The panels from left to right have different 
values of perturbation strength ratio $\mu/\epsilon$.    Initial conditions for the orbits are drawn from 
uniform distributions in phase space. 
Each orbit is a set of points of the same 
color and they are plotted in phase space $\phi,p$.  The parameters of the Hamiltonian
are written on the plot except for $\tau_0$ which is set to 0. 
b)  Similar to a) but showing Poincar\'e maps for the case $\lambda =1/3$ and the same values of 
the ratio  $\mu/\epsilon$ as in a).   Comparison between a) and b) illustrates that 
when the perturbation frequency is near but somewhat 
 lower than the characteristic libration frequency, there are fewer and smaller resonant islands.  
We have chosen strong perturbation parameters $\mu, \mu'$ so that the ergodic regions 
cover phase space. 
\label{fig:SS}
 }
\end{figure*}

It has been conjectured \citep{Bohigas_1984} that the eigenvalue statistics of a quantum system derived from a classical chaotic one can be described by random matrix theory, whereas the eigenvalue statistics of a quantum system 
derived from a classical integrable system follows Poisson statistics \citep{Berry_1977}. 
Though a chaotic classical system, when quantized, can exhibit properties similar to 
matrices chosen from 
a random matrix model (e.g., \citet{Izrailev_1989}), 
a single chaotic classical system, when quantized, would give only a single
unitary operator, not a sequence of them.  
Hence a single chaotic system is not by itself a quantum sampler. 

A classical chaotic system exhibits sensitivity to initial conditions, with nearby orbits diverging rapidly away 
from each other, as shown in Figure \ref{fig:illust}a.  If the chaotic system dynamics is a function of some parameters, orbits can 
be extremely sensitive to these parameters (e.g., \citep{Benettin_1984}).  In other words if the system is 
chaotic and its dynamics is sensitive to a parameter $\mu$, 
then two orbits could diverge exponentially if they have the same initial conditions but are undergoing evolution with different values of $\mu$ (see Figure \ref{fig:illust}b).    
Quantization of the  classical chaotic system can give a unitary operator (a propagator) 
$\hat U(\mu)$ that is also a function of the parameter $\mu$. 
Sensitivity of orbits to the parameter $\mu$ in the  classical system implies that
transition probabilities between states in the associated quantum system are also sensitive to the 
parameter $\mu$, as shown in Figure \ref{fig:illust}c.
Small variations in the parameter $\mu$ could give large variations in this unitary operator. 
Even though the group of unitary operators U$(N)$ is compact, 
its dimension is not necessarily small as each unitary matrix in U$(N)$ is described 
by $N^2$ elements.  If $\hat U(\mu)$ is strongly dependent upon 
$\mu$ then the curve $\hat U(\mu)$  in U$(N)$ generated from an interval of possible $\mu$ values 
can wander in U$(N)$ (e.g., \citep{Nakata_2017}, see the illustration 
in Figure \ref{fig:UN}).   With a probability distribution for 
values of $\mu$ we would generate a distribution of unitaries along the curve $\hat U(\mu)$.  

If instead of choosing a single real parameter $\mu$ to describe both classical 
and associated quantum systems, a vector of control parameters $\boldsymbol \mu$ could describe both classical and quantum systems.  
For example if $\boldsymbol \mu$  is a 2-dimensional vector, 
 the set  of operators $\hat U({\boldsymbol \mu})$ would give a two dimensional 
surface embedded in U$(N)$.   
A distribution of parameters $\boldsymbol \mu$ then generates a distribution 
of unitaries. 

In a classical ergodic system, the maximal Lyapunov exponent describes the rate that 
two initial trajectories diverge.  A related notion in a quantum system 
is its sensitivity to perturbations \citep{Feingold_1986,Scott_2006}.  
Ergodic behavior in a quantum system can be described in terms of how fast a statevector,  
evolving according to the original system, diverges from that of the same statevector evolving by a perturbed version of the quantum system \citep{Scott_2006}. 
Instead of considering the divergence of two states as a function of time (e.g., \citep{Roberts_2017}), 
we consider how  two unitary operators diverge from each other as a function of 
a parameter that is used to describe the strength of the perturbation. 
With a sufficient number of parameters and sufficient sensitivity of the propagator to these parameters, 
the distribution of randomly chosen unitaries would be sufficiently well distributed
throughout U$(N)$ that the resulting distribution is an approximate t-design and approximates 
a Haar random distribution of unitaries. 


\section{A chaotic Floquet system on the torus}

Our goal is to use a chaotic unitary operator 
that depends upon a few parameters 
to generate a series of randomly generated unitaries
that approximate a Haar-random generated set.    
We begin by studying the properties of unitaries generated from a chaotic finite dimensional Floquet 
system on a torus.  

The Floquet system we study is that characterized by \citep{Quillen_2025}, and is the Hamiltonian operator 
\begin{align}
\hat h(\tau) = \hat h_0 + \hat h_1(\tau) \label{eqn:hath}
\end{align}
in a finite dimensional Hilbert space of dimension $N$,  
where $\hat h_0$ is time independent and $\hat h_1(\tau)$ we call the perturbation, even though 
it need not necessarily be small.  We assume that 
the time dependent perturbation $\hat h_1(\tau)$ is periodic with period $T = 2 \pi/\nu$. 
We set the 
 perturbation frequency $\nu=1$, so that time is in units of the inverse of the perturbation frequency. 
The time independent portion of the Hamiltonian is a shifted version of the Harper model 
\begin{align}
\hat h_0 = a(1- \cos (\hat p-b)) - \epsilon \cos (\hat \phi  - \phi_0). \label{eqn:hath0}
\end{align}
In a basis $B_{\hat \phi}$ denoted by the set $\{ \ket{j}   \}$ with $j \in \{ 0, 1, \ldots, N-1 \}$, the angle 
 operator is diagonal; 
\begin{align}
\hat \phi = \sum_0^{N-1} \frac{2 \pi j}{N} \ket{j}\bra{j} . 
\end{align}
The momentum operator is related to the angle operator via a Discrete Fourier transform
 $\hat p = Q_{FT} \hat \phi Q_{FT}^{-1}$, with 
$\hat Q_{FT} =\frac{1}{\sqrt{N}} \sum_{jk} \omega^{jk}  \ket{j}\bra{k} $ 
and $\omega = e^{\frac{2 \pi i}{N}}$. 
The parameters   
 $a, \epsilon, b, \phi_0$ are real numbers and $a, \epsilon$ are usually taken to be positive. 
 Whereas \citet{Quillen_2025} considered $\hat h_0$ and $\hat h_1$ with $b=\phi_0=0$, 
\citet{Quillen_2025b} included the parameter $b$ in the kinetic term.  Here we also include 
a parameter $\phi_0$ which allows us to shift the phase of the potential term. 

The associated classical Hamiltonian system is similar to the quantum system but 
with operator $\hat p$ replaced by momentum $p$ and operator $\hat \phi$ replaced by 
 angle $\phi$; 
\begin{align}
H(p,\phi,\tau) = H_0(p,\phi) + H_1(\phi,\tau) \label{eqn:Hclass}
\end{align}
and unperturbed Hamiltonian 
\begin{align}
H_0(p,\phi) = a(1- \cos (p-b)) - \epsilon \cos  (\phi - \phi_0)  \label{eqn:H0}.
\end{align} 
Phase space is doubly periodic $p, \phi \in [0,2\pi)$ and equivalent to a torus.  
In the classical system, the parameter $b$ simply shifts the momentum. However if $b$ 
is not a multiple of $2 \pi/N$  then the spectrum of the operator $\hat h_0$ differs from that with $b=0$ 
  \citep{Quillen_2025b}. 
  
\citet{Quillen_2025} considered periodic perturbations motivated by 
the periodically perturbed pendulum \citep{Chirikov_1979} 
\begin{align}
\hat h_1(\tau)  = &-\mu \cos (\hat \phi  - \phi_0 + \tau - \tau_0)  \nonumber \\
& \ \ \  -\mu' \cos (\hat \phi - \phi_0 - \tau + \tau_0). \label{eqn:hath1}
\end{align}
The associated classical system, the Hamiltonian perturbation 
$H_1(\phi,\tau) $ has the same form as equation \ref{eqn:hath1} 
\begin{align}
\hat H_1(\tau) = &  -\mu \cos (\phi - \phi_0 + \tau - \tau_0)  \nonumber \\
& \ \ \  -\mu' \cos (\phi - \phi_0- \tau + \tau_0). \label{eqn:H1}
\end{align}
The perturbation contains two terms, one with amplitude $\mu$, the other with amplitude $\mu'$. 
The angle $\tau_0$ sets the initial phase of perturbation. 
The perturbation causes both classical and quantized system to become chaotic near the separatrix orbit 
which divides librating and circulating orbits in the classical system \citep{Quillen_2025}. 

\subsection{Sensitivity of the extent of the ergodic region to the ratio of perturbation to libration frequency}
\label{sec:lambda}
The characteristic frequency of libration at the stable fixed point of the unperturbed
classical Hamiltonian $H_0$ (equation \ref{eqn:H0})  is 
\begin{align}
\omega_0 = \sqrt{a\epsilon}. \label{eqn:omega0}
\end{align}
Because the perturbation frequency  $\nu=1$, the ratio of perturbation frequency to 
characteristic frequency of libration is 
\begin{align}
\lambda =\frac{\nu}{\omega_0} =  \frac{1}{\sqrt{a \epsilon}}.
\end{align}

\begin{figure*}[!htb]\centering 
\includegraphics[width=3truein]{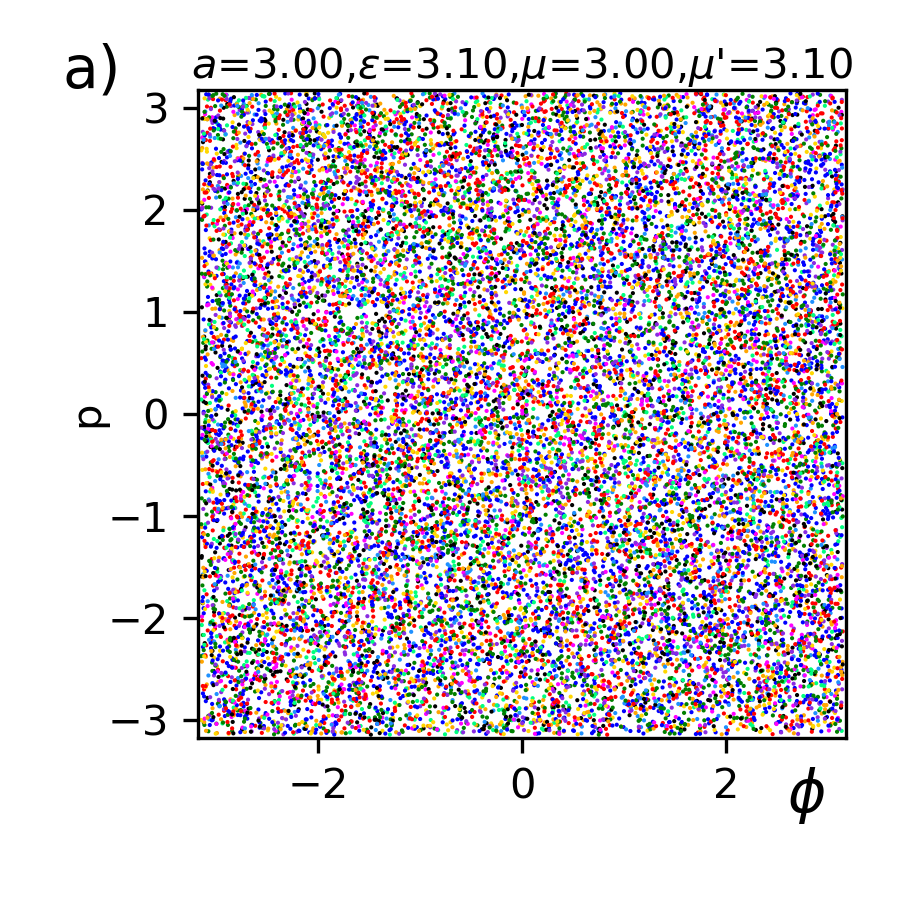}
\includegraphics[width=3truein]{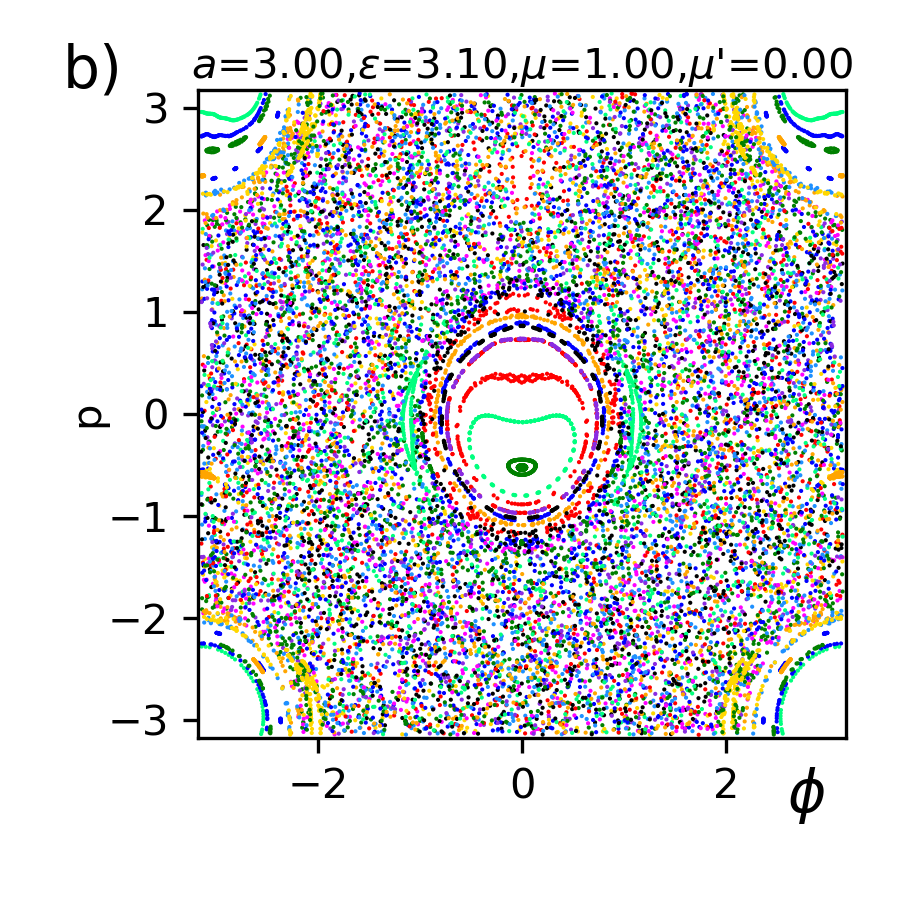}
\caption{a) Poincar\'e maps of the associated classical 
models for the Floquet propagator $\hat U_{Ta}$ with properties listed in Table \ref{tab:props}. 
The ergodic region covers phase space.  
b) similar to a) but for propagator $\hat U_{Tb}$.   
There are both integrable and ergodic regions in phase space. 
The Poincar\'e maps are similar to those shown in Figure \ref{fig:SS} but for different classical Hamiltonians.  
 \label{fig:SSab}}
\end{figure*}

The ratio of perturbation frequency to libration frequency $\lambda$
affects estimates for the extent of chaos induced by the perturbation 
\citep{Zaslavsky_1968,Chirikov_1979,Shevchenko_2000,Soskin_2008,Soskin_2009,Treschev_2010}.   
If $\lambda \ll 1$ the perturbation is called adiabatic. 
With $\mu/\epsilon, \mu'/\epsilon$ fixed, the width of the chaotic region in phase space is
maximized where the perturbation is neither adiabatic nor fast at $\lambda \sim 1$.  
However, 
with $\lambda $ near but less than 1, the chaotic region is wide and there are fewer resonances that cause 
substructure near the separatrix and variations in the width of the chaotic region \citep{Soskin_2009}. 
The tendency for the regime with $\lambda >1$ to exhibit resonant islands is illustrated in Figure \ref{fig:SS}. 
In Figure \ref{fig:SS} we illustrate Poincar\'e maps, also called surfaces of section, which are orbits with 
  points plotted every perturbation period of $2\pi$.  Figure \ref{fig:SS}a shows three
  surfaces of section all with $\lambda = 2$, whereas  Figure \ref{fig:SS}b shows a similar set 
  but with $\lambda = 1/3$.  From left to right the perturbation strengths increase in both subfigures. 
 Figure \ref{fig:SS} illustrates that with $\mu \sim \mu' \sim \epsilon$, the chaotic region covers phase space  
 and that with $\lambda$ near but less than 1 there are fewer resonant islands. 
 
Resonant islands are integrable regions of phase space.  When perturbed, the orbits could 
simply be slightly perturbed.  Eigenstates of the 
Floquet propagator for the associated quantized system that are in the same region of phase space
would not be expected to exhibit extreme sensitivity to perturbation.  
To exploit the chaotic nature of the classical system and its quantum counterpart, 
we desire a chaotic region that covers phase space and lacks resonant islands.     
With $\lambda $ near 1, the perturbation is neither adiabatic nor rapidly varying compared
to the libration period at the bottom of the potential well,  and 
the chaotic region has a maximum width \citep{Chirikov_1979}.  
With $\lambda $ less than 1, the ergodic region would lack resonant substructure. 
Hence we desire a system with ratio $\lambda$ near but somewhat less than 1. 


 
\begin{table}[!htb]
\caption{Propagators \label{tab:props}}
\begin{ruledtabular}
\begin{tabular}{llll}
Propagator & Parameters of $\hat h$\footnote{
The propagators $\hat U_{Ta}, \hat U_{Tb}$ and $U_\text{Drift}$ are generated 
from the Hamiltonian operator $\hat h$ described in 
equations \ref{eqn:hath}, \ref{eqn:hath0}, and \ref{eqn:hath1}.  The parameters
listed in the right column are those of this Hamiltonian. 
Parameters $\phi_0=\tau_0=0$ for the first three of these propagators. 
For $\hat U_\text{Drift}$ the time derivatives of parameters $b, \mu, \mu'$ are constant and we list
initial and final values for these drifting parameters. 
In Figure \ref{fig:Hus} the dimension of the operators $N=49$ and chosen to be a square so 
that the images fill a square. 
In Figure \ref{fig:spacing}, we chose $N=600$ as a large dimension is required to sample a distribution of 
eigenvalue spacings.  In Figures \ref{fig:tt}, \ref{fig:PT}, \ref{fig:IPR}, \ref{fig:Wvals}, and \ref{fig:wcumu}, 
and in Tables \ref{tab:mom_trans} and \ref{tab:mom_ops} the dimension $N=81$.}\\
\colrule
$\hat U_{Ta}$  &  $a=3,b=0.2,\epsilon=3.1,\mu=3,\mu'=3.1$ \\
$\hat U_{Tb}$  &  $a=3,b=0.0,\epsilon=3.1,\mu=1,\mu'=0$ \\
$\hat U_\text{Drift}$ &  $a=3, b=0 \to 1.9, \epsilon=3$ \\
                      & $\mu = 1 \to 7, \mu' = 0.5 \to 6.5, n_\text{periods}=3$ \\
$\hat U_\text{Haar}$ & 1 unitary matrix generated \\
                    &    via a pseudo-random generator  \\
\end{tabular}
\end{ruledtabular}
\end{table}

\begin{figure*}[htbp]\centering 
\includegraphics[width=3truein]{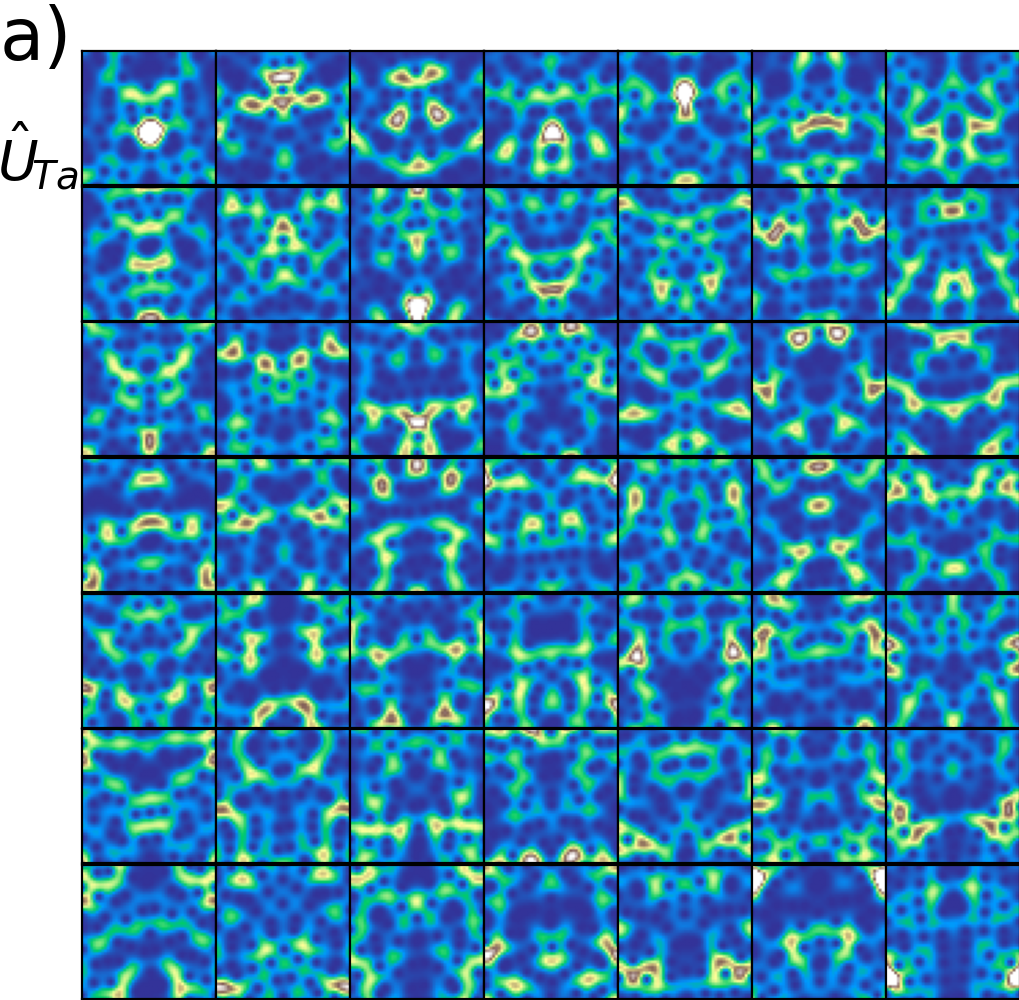 }
\includegraphics[width=3truein]{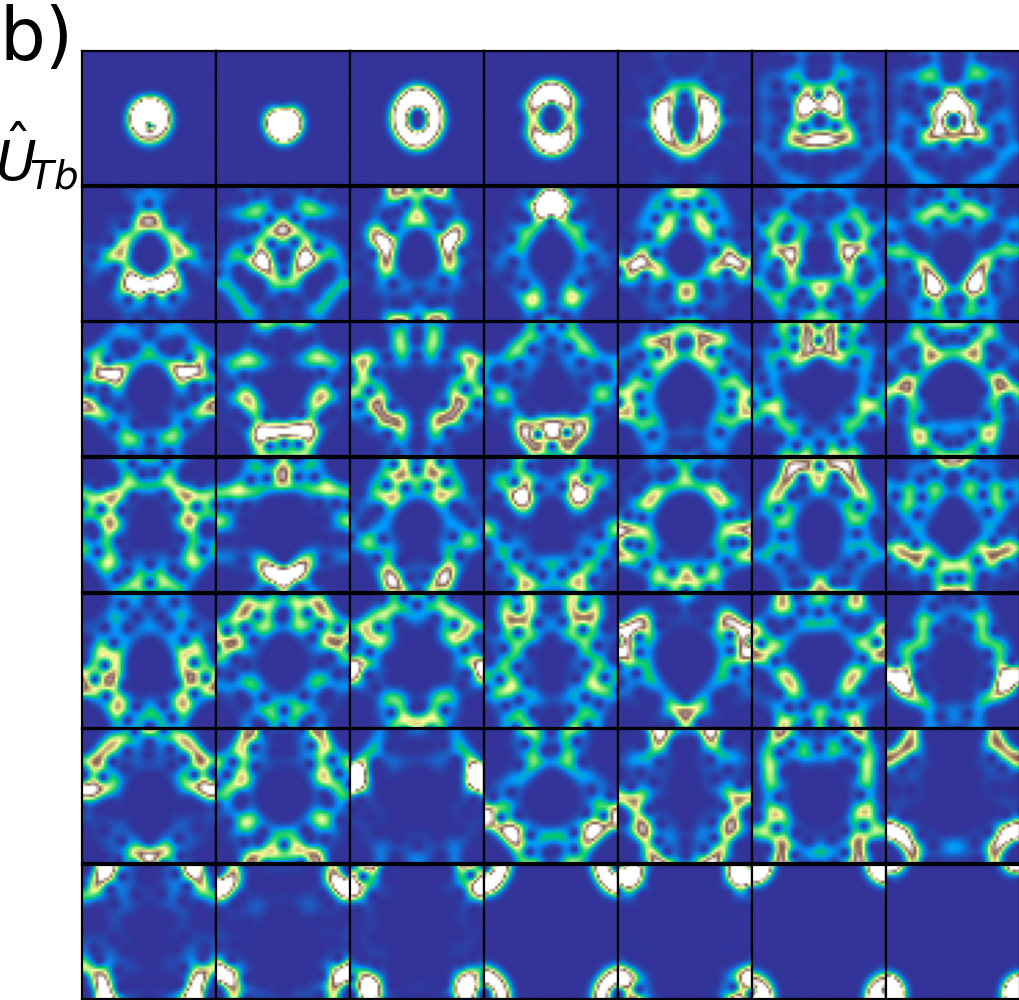 }
\includegraphics[width=3truein]{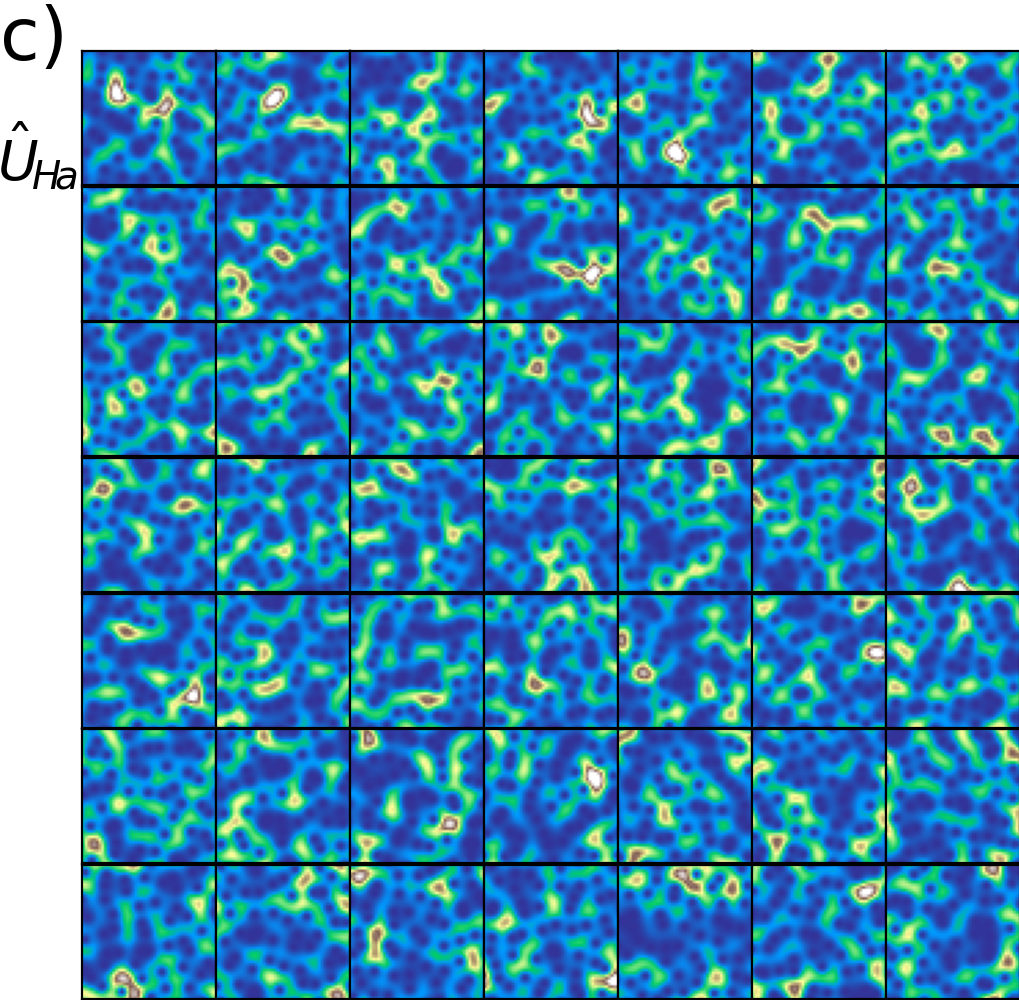 }
\includegraphics[width=3truein]{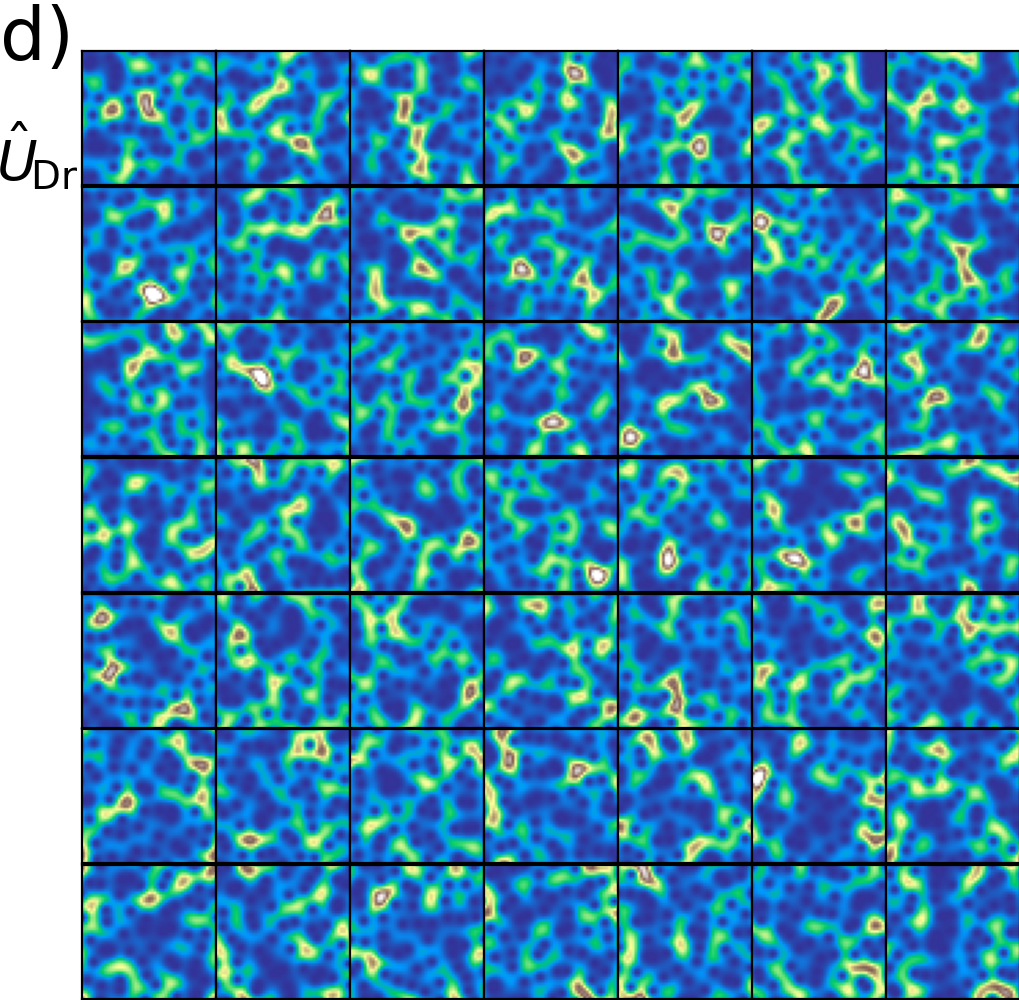 }
\caption{a) Husimi distributions of the eigenvectors of  the Floquet propagator $\hat U_{Ta}$
with parameters listed in Table \ref{tab:props} and dimension $N=49$.  
Each panel shows the Husimi distribution for a single eigenvector in phase space. 
The classical version of this operator is ergodic; see the Poincar\'e map in Figure \ref{fig:SSab}a. 
Most of the eigenfunctions are distributed across phase space.   
b)  The Husimi distributions for $\hat U_{Tb}$ which is hybrid in the sense that there 
are both chaotic and integrable regions in phase space.    
There are Husimi distributions  resembling non-chaotic localized orbits in the Poincar\'e map of Figure \ref{fig:SSab}b.  
c) The Husimi distributions are generated from the randomly generated unitary operator $\hat U_\text{Haar}$ and are widely distributed in phase space. 
d) The Husimi distributions are generated from $\hat U_\text{Drift}$ which is the drifting 
system with parameters listed in Table \ref{tab:props}.  Husimi distributions are also 
widely distributed in phase space. 
\label{fig:Hus}}
\end{figure*}

\section{Comparing the properties of 4 propagators}
\label{sec:comp4}

A propagator is generated from the operation of the time dependent Hamiltonian operator 
\begin{align}
\hat U_T = {\cal T} e^{ -i \frac{N}{2\pi}\int_0^{T} \hat h(\tau) d\tau  } \label{eqn:UT}
\end{align}
integrated over time duration $T$. 
Here $\cal T$  indicates a product of time-ordered operators computed in the limit of small step size.
The factor of $N/2\pi$ is consistent with an effective Planck's constant of 
\begin{align}
 \tilde \hbar = \frac{2 \pi}{N}
 \end{align} 
 which arises during quantization of the Harper Hamiltonian (see section II A by \cite{Quillen_2025}). 
We compute propagators for both Floquet and drifting systems  via 
Suzuki-Trotter decomposition, following the procedure described in appendix D by \citet{Quillen_2025}.  
The number of Trotterization steps per $T=2\pi$ perturbation period used is $n_\tau = 4N$.  
 

Floquet systems are those with a periodic Hamiltonian operator.  
We can modify the 
Floquet system with Hamiltonian operator in equations \ref{eqn:hath}, \ref{eqn:hath0}, \ref{eqn:hath1} 
and propagator in equation \ref{eqn:UT}  so that its parameters are slowly drifting. 
We refer to the periodic systems as Floquet systems and those that drift as drifting systems. 
For a Floquet system, 
the unitary operator $\hat U_T$ describes the evolution of the quantum
system over a single perturbation period; from dimensionless time $\tau = 0$ to $2 \pi$. 
For the drifting systems, we describe the drift duration $n_\text{periods}$  in units of the $2 \pi $ perturbation period.  

We compare the properties of 4 unitary operators that are listed in Table \ref{tab:props} 
along with parameters used to describe them.
The unitary operator $\hat U_{Ta}$ is a Floquet system that is chosen to match a classical 
system that is uniformly ergodic.  To ensure that the ergodic region covers phase space
without resonant islands (see section \ref{sec:lambda}) we 
chose ratio of perturbation to libration frequency near but less than one. 
We chose parameters $a \ne \epsilon$ and $b\ne 0$ so that the time independent 
portion $\hat h_0$ of the Hamiltonian obeys fewer symmetries (for a list of symmetries obeyed by $\hat h_0$ 
see the appendices by \citealt{Quillen_2025b}). 
The Hamiltonian of Floquet propagator $\hat U_{Tb}$ has a lower perturbation strength than $\hat U_{Ta}$.
Poincar\'e maps generated from the classical Hamiltonian versions associated with 
 $\hat U_{Ta}$ and  $\hat U_{Tb}$ are shown in Figure \ref{fig:SSab} and illustrate 
 that  $\hat U_{Ta}$ is fully ergodic and  $\hat U_{Tb}$ is a hybrid system with both ergodic 
 and integrable regions in phase space. 
 The Floquet operators depend upon parameters $a,b,\epsilon, \mu, \mu', \phi_0, \tau_0$ 
 in equations \ref{eqn:hath}, \ref{eqn:hath0}, and \ref{eqn:hath1} and the values for these parameters for 
 the Hamiltonians used to generate $\hat U_{Ta}$  and $\hat U_{Tb}$ are listed in Table \ref{tab:props}. 

The unitary operator $\hat U_\text{Drift}$ is a drifting system. 
It is similar to the other systems except that its parameters slowly vary. 
For this case we compute the propagator 
over $n_\text{periods}$ = 3 perturbation periods.  
For the drifting system in Table \ref{tab:props} we list initial and final values of the parameters $a,b,\epsilon, \mu, \mu', \phi_0, \tau_0$ when they are varied otherwise we list their fixed values. 
The system is varied so that these parameters have time derivatives that are constant. 

As a control case we compute a unitary matrix using 
a pseudo-random unitary generator (python's routine \texttt{scipy.stats.unitary\_group}) that generates
unitaries from a Haar uniform distribution.  
The computed unitary we  denote  $\hat U_\text{Haar}$.  
Notes in Table \ref{tab:props} list the dimension $N$ of the quantum space of the operators 
used in the different figures and tables. 

\subsection{Husimi distributions}

Husimi distributions (also called the Husimi Q representation),
are quantum analogs to classical phase space distributions  \citep{Cartwright_1976}. 
For the perturbed Harper model, 
Husimi distributions created from eigenstates of a Floquet propagator $\hat U$  
resemble orbits of the associated classical model \citep{Quillen_2025}.  
Following the procedure outlined in appendix B by \citep{Quillen_2025}, we construct 
Husimi distributions for each eigenstate of each unitary operator listed in Table \ref{tab:props}
with dimension $N=49$.  This dimension is chosen so that the Husimi distributions are easily displayed 
as they fill a square.  The resulting distributions are shown in Figure \ref{fig:Hus}. 

In Figure \ref{fig:Hus}a there are $N=49$ panels, one for each eigenstate of $\hat U_{Ta}$.  
The image in each panel 
represents probability for a particular eigenstate as a function of position in phase space, with horizontal direction $\phi$ and vertical 
direction $p$.  The colormap and scale are identical in every panel and in every subfigure. 
In the chaotic regions of phase space 
 in the associated classical model shown in the surfaces 
of section shown in Figure \ref{fig:SSab}a,b, 
the Husimi distributions in Figure \ref{fig:Hus}a,b are widely distributed in phase space.  
Because the ergodic region for $\hat U_{Ta}$ covers phase space, most of the Husimi distributions 
in Figure \ref{fig:Hus}a are widely
distributed. In contrast only part of phase space is ergodic for $\hat U_{Tb}$, and in Figure \ref{fig:Hus}b there are Husimi distributions that resemble the non-chaotic orbits present in 
the Poincar\'e map of Figure \ref{fig:SSab}b. 
The Husimi distributions in Figure \ref{fig:Hus}c, constructed from eigenstates of 
a  particular randomly chosen unitary matrix, and they all 
 appear uniformly distributed in phase space. 
Similarly the Husimi distributions constructed from eigenstates of the drifted operator $\hat U_\text{Drift}$,  
 shown Figure \ref{fig:Hus}d,  appear uniformly distributed. 

\begin{figure}[htbp]\centering
\includegraphics[width=3.0truein]{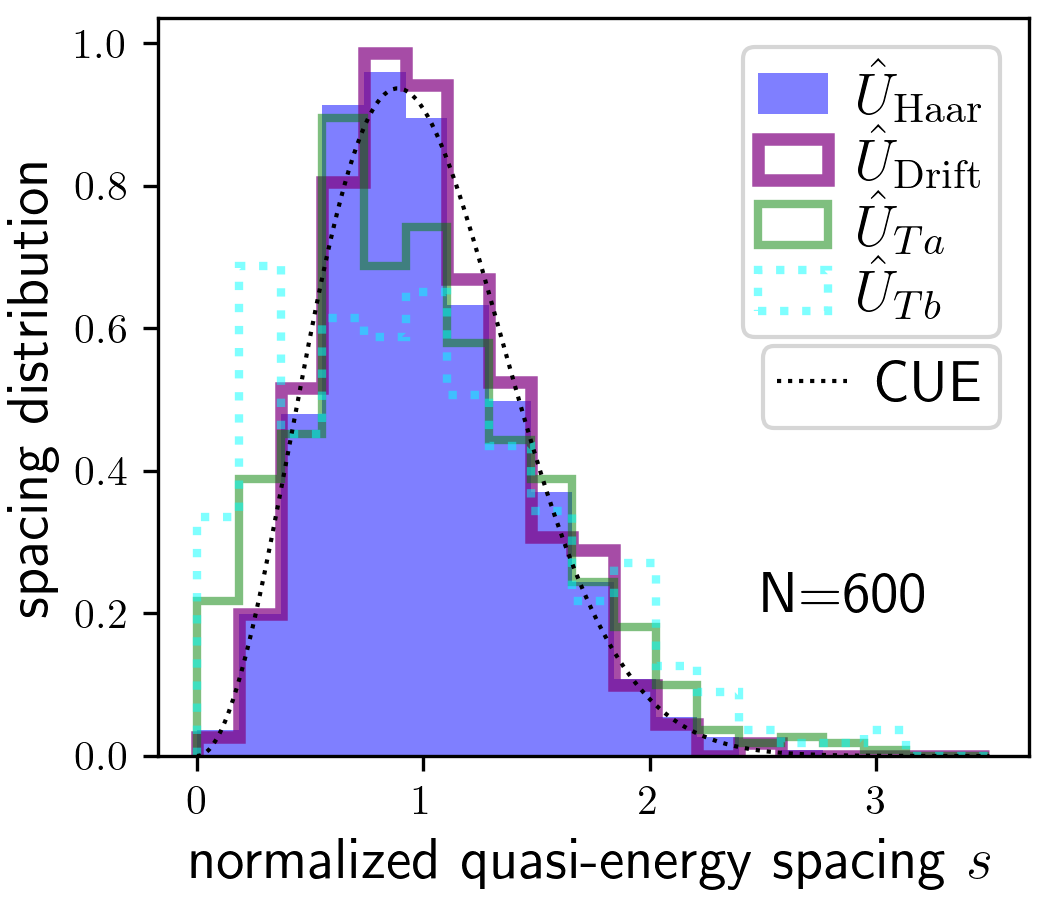}
\caption{Quasi-energy spacing histograms computed for the 4 propagators listed in Table \ref{tab:props} 
with $N=600$. Here $s$ is the spacing divided by the mean value. 
We also show the spacing distributed predicted for the circular unitary ensemble (CUE, equation \ref{eqn:CUE}) as a thin black dotted line.  The Haar random and drifting propagators $\hat U_\text{Haar}$ and $\hat U_\text{Drift}$
exhibit spacings in their quasi-energies that are closer to that of the CUE than the two Floquet propagators 
$\hat U_{Ta}, \hat U_{Tb}$. 
\label{fig:spacing} }
\end{figure}

\subsection{Quasi-energy spacings}

According to the Bohigas, Giannoni and Schmit conjecture \citep{Bohigas_1984},   
 the statistics of the eigenphases or quasi-energies of a Floquet propagator that is 
 related to a classically chaotic system would be similar to those of matrices 
 chosen from a random matrix ensemble. 
As the propagators are unitary, the quasi-energy 
 statistics would be described by the circular unitary ensemble or CUE. 
The CUE has quasi-energy spacing distribution 
\begin{align}
p_\text{CUE}(s) = \frac{32}{\pi^2} e^{- \frac{4s^2}{\pi}} \label{eqn:CUE}
\end{align}
(via Wigner's surmise) 
where the spacing $s$ between the eigenphases is normalized by mean value of the entire set of spacings. 

In Figure \ref{fig:spacing}
quasi-energy spacing distributions are displayed for the 4 unitaries listed in Table \ref{tab:props}. 
We use a larger value of $N$ (than in Figure \ref{fig:Hus}) to compute the 
unitaries so that we can better see the distribution. 
As expected, the spacing distribution of the 
ergodic drifted $\hat U_\text{Drift}$ unitary resembles that of the Haar-randomly chosen 
matrix $\hat U_\text{Ha}$, and both of these resemble the distribution predicted for the CUE.  
The spacing distribution of the hybrid operator $\hat U_{Tb}$ deviates from the CUE's distribution. 

A set of ordered energy levels can be characterized with an average of ratios  
that depend on consecutive spacings or gaps  \citep{Oganesyan_2007,Atas_2013,DAlessio_2014,Lazarides_2015,Ponte_2015}.
Following \citet{Oganesyan_2007}, 
the energies are sorted so that they are in order of increasing value.  We define   
 $s_j$  to be the difference between the $j$-th  and $(j+1)$-th energies after sorting them.  
We compute a ratio $r_j$ for each consecutive gap  
\begin{align}
0 \le r_j = \frac{ \text{min} (s_j, s_{j+1}) }{ \text{max} (s_j, s_{j+1}) } \le 1 . \label{eqn:r_j}
\end{align} 
A disordered and delocalized Floquet system has average $r$-ratio near that of a random matrix 
 ensemble and one that is ordered or localized would have a lower average $r$-ratio 
 closer to that of the Poisson distribution \citep{Ponte_2015}. 

For the 4 unitaries listed in Table \ref{tab:props} computed with $N=600$ and shown in Figure \ref{fig:spacing}
 we computed the $r$-ratios and list the average values in Table \ref{tab:rratio}. 
 We use the same propagators that have quasi-energy spacings shown in Figure \ref{fig:spacing}.  
For comparison we also list in Table \ref{tab:rratio} the average $r$-ratio values  for the CUE and Poisson 
distributions estimated by \citet{DAlessio_2014}. 
As expected from the shape of the spacing distributions,  the Haar-randomly chosen 
matrix $\hat U_\text{Ha}$ and ergodic drifted $\hat U_\text{Drift}$ unitary  have average $r$-ratio close to that predicted for the CUE.  The ergodic Floquet unitary $\hat U_{Ta}$ has a lower value 
and the hybrid Floquet unitary $\hat U_{Tb}$ has an even lower value of the average $r$-ratio, 
though both are significantly above the $r$-ratio average of a Poisson ensemble. 

\begin{table}[!htb]
\caption{Average ratio of adjacent quasi-energy gaps \label{tab:rratio}}
\begin{ruledtabular}
\begin{tabular}{llll}
Propagator & $r$-ratio average\footnote{
The $r$-ratios (defined in equation \ref{eqn:r_j}) are computed for
the 4-propagators listed in Table \ref{tab:props} for dimension $N=600$.  The same propagators
are used to make Figure \ref{fig:spacing} showing the normalized quasi-energy gaps.}\\
\colrule
$\hat U_{Ta}$  & 0.532  \\
$\hat U_{Tb}$  & 0.501   \\
$\hat U_\text{Drift}$ &  0.604  \\
$\hat U_\text{Haar}$ & 0.614 \\
\hline 
CUE\footnote{The estimated average values for a the CUE and Poisson distributions are from \citet{DAlessio_2014}.}  & 0.59643 \\
Poisson$^\text{b}$ &  0.386294 \\
\end{tabular}
\end{ruledtabular}
\end{table}

\subsection{Transition probabilities}
\label{sec:trans}

\begin{figure*}[htbp]\centering
\includegraphics[width=5.0truein]{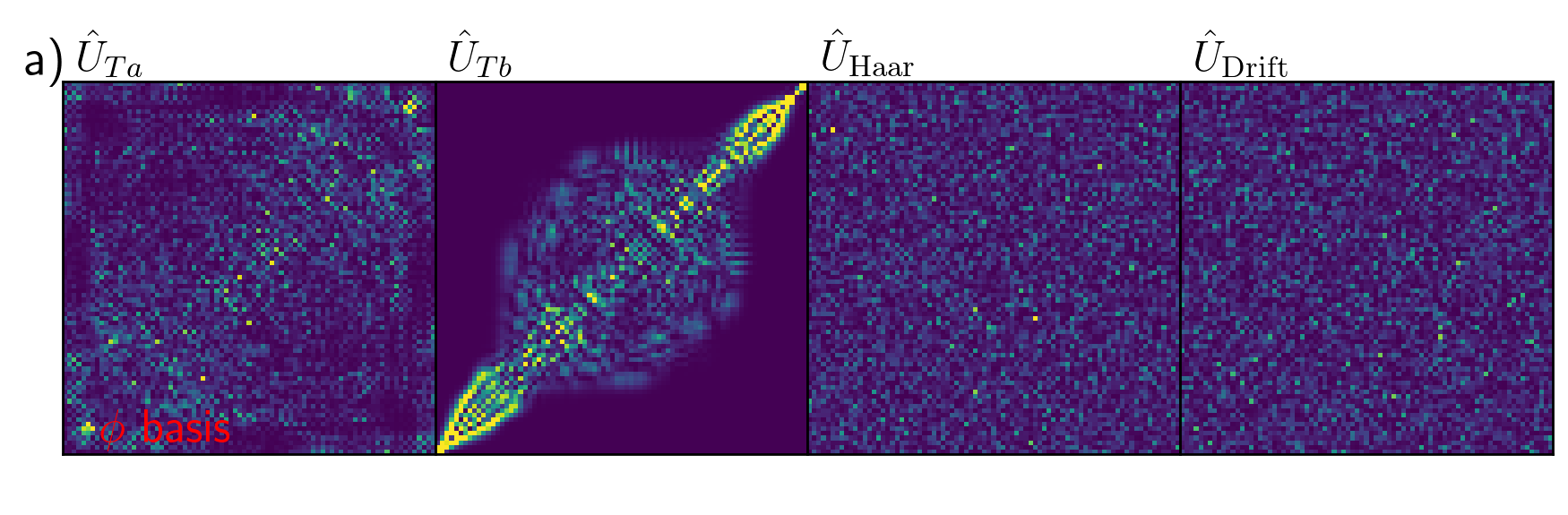}
\includegraphics[width=5.0truein]{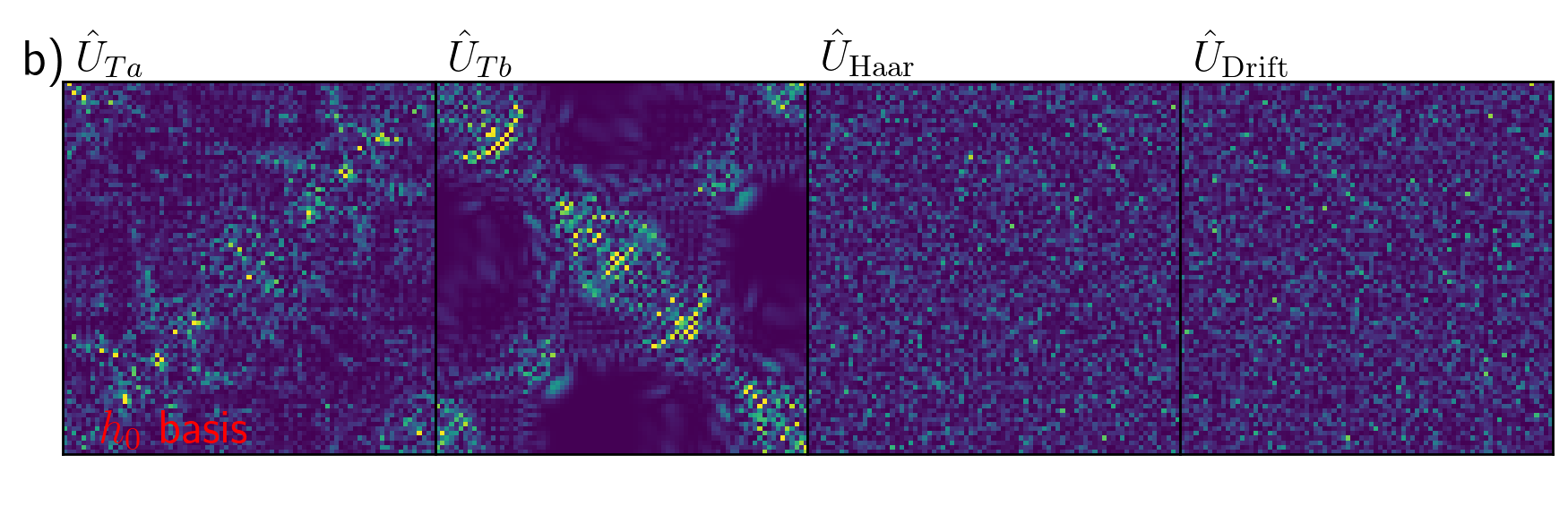}
\caption{a) Transition probabilities $|\bra{j} \hat U \ket{k}|^2$ for all $j,k$ are shown 
as images  for the 4 propagators listed in Table \ref{tab:props} with $N=81$.  
The transition probabilities are computed in the orthonormal basis $B_{\hat \phi}$ 
which is associated with eigenvectors of the angle operator $\hat \phi$. 
b) similar to a) but in computed in the eigenvector basis of $\hat h_0$. 
\label{fig:tt} }
\end{figure*}

\begin{figure*}[htbp]\centering
\includegraphics[width=1.6truein]{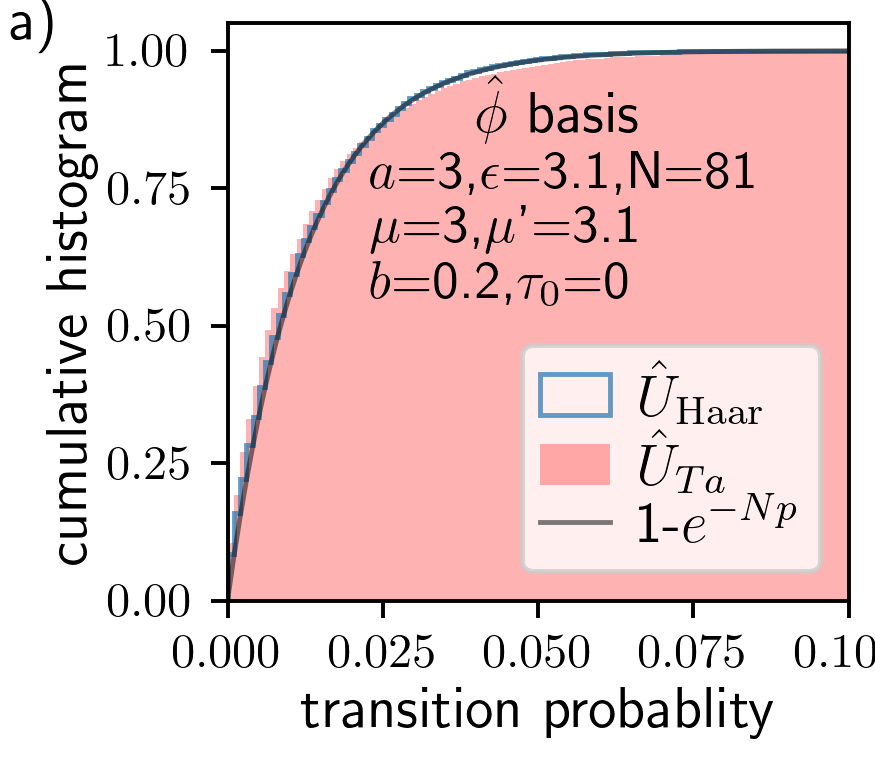}
\includegraphics[width=1.6truein]{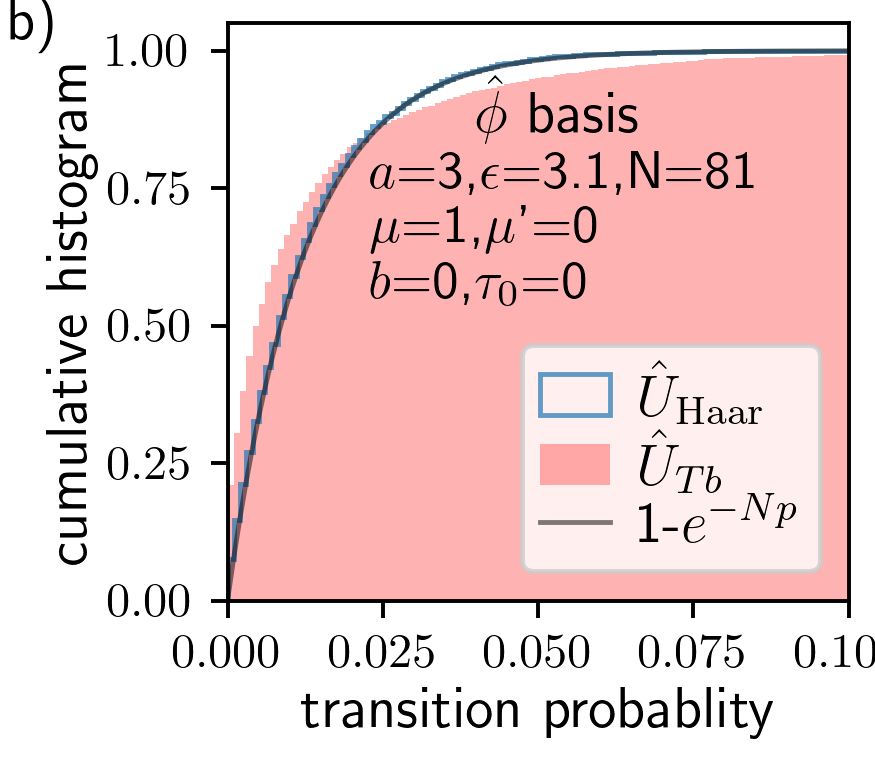}
\includegraphics[width=1.6truein]{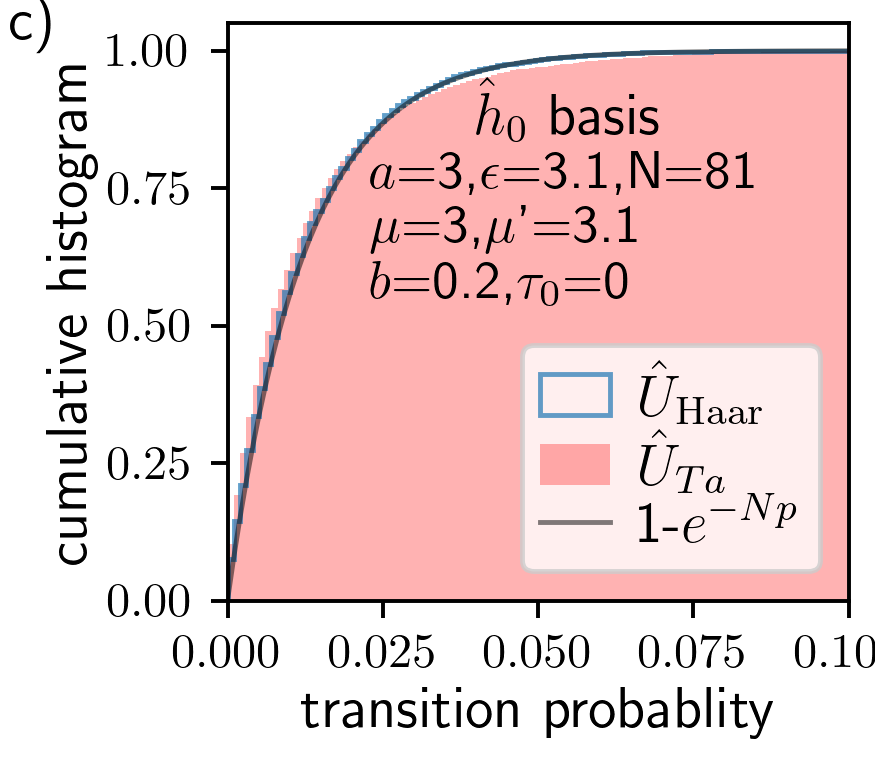}
\includegraphics[width=1.6truein]{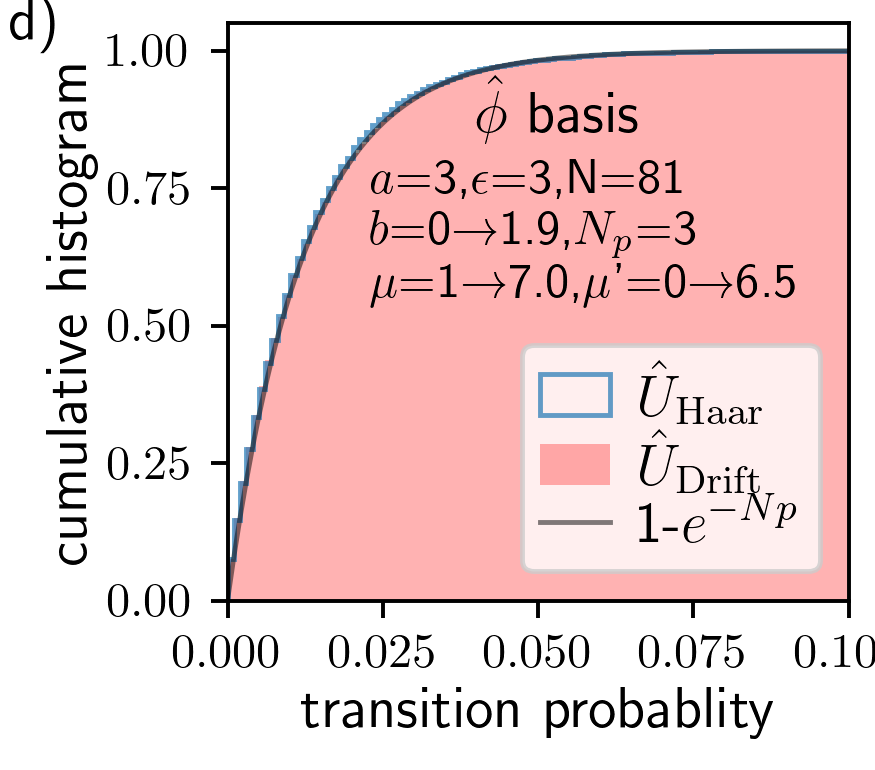}
\caption{a) With a filled pink histogram we show the cumulative distribution of transition probabilities $|\bra{j} \hat U_{Ta} \ket{k}|^2$ for all $j,k \in   {\mathbb Z}_N  $ (equation \ref{eqn:transprobs}).  The Floquet propagator $\hat U_{Ta}$ is chaotic and generated from the Hamiltonian in equations \ref{eqn:hath}, \ref{eqn:hath0}, and \ref{eqn:hath1} and has 
parameters shown on the plot and listed in Table \ref{tab:props}.   
The basis used is the $\{ \ket {j} \}$ basis of the angle $\hat \phi$  operator.  
The black line shows the function $1-e^{-Np}$ which gives the distribution 
expected for a random unitary.   The cumulative histogram generated from 
a particular randomly generated unitary $\hat U_{Ha}$ is shown with a blue line.  
The histogram for $\hat U_{Ta}$ is close to  that of the randomly generated unitary. 
b) Similar to a) but for the propagator $\hat U_{Tb}$ with a smaller amplitude perturbation so the system contains a smaller chaotic  region in phase space.  
c) Similar to a) but the transition probabilities are computed using a basis of eigenvectors of $\hat h_0$ 
(equation \ref{eqn:hath0}) the unperturbed Hamiltonian with parameters  $a,b,\epsilon$, the same
as for $\hat U_{Ta}$. 
d) Similar to a) but for a drifting system.  The propagator is $\hat U_\text{Drift}$ with parameters 
listed in Table \ref{tab:props}. 
 \label{fig:PT}
}
\end{figure*}

Anti-concentration or delocalization describes a state that is spread out in a particular basis. 
We take basis $B : \{ \ket{v_j} : j \in  {\mathbb Z}_N  $ to be a set of 
 linear independent and orthonormal vectors ($\bra{v_i}\ket{v_j} = \delta_{ij}$) that spans 
 the quantum vector space. 
Here ${ \mathbb Z}_N $ is short-hand for the set $j = \{ 0 , 1, \ldots, N-1 \}$. 

The eigenvectors of a Hermitian operator can be used to construct such a basis. 
We denote basis $B_{\hat \phi}$ as a basis in which the angle operator is diagonal; 
\begin{align}
B_{\hat \phi} \equiv  \left\{  \ket{j}  : j \in { \mathbb Z}_N \right\}   .
\end{align} 
We denote basis $B_{\hat h_0}$ as a set of normalized eigenvectors 
of the operator $\hat h_0$ (the unperturbed portion of the Hamiltonian in equation \ref{eqn:hath0}). 

To illustrate that an operator $\hat U$ is anti-concentrated or delocalized with respect 
to an orthonormal basis $B$,  we can compute transition probabilities 
for a pair of states in the basis; 
\begin{align}
z_{jk}(B,\hat U ) = |\bra{v_j} \hat U \ket{v_k}|^2  \text { with } \ket{v_j}, \ket{v_k} \in B. 
\label{eqn:transprobs}
\end{align}
Here $z_{jk}$ is the probability that a state $\ket{v_k}$ (in orthonormal basis $B$), after evolving according to 
the unitary operator $\hat U$, is 
measured to be in state $\ket{v_j}$ (also in basis $B$). 
We call the set of transition probabilities 
\begin{align}
S_\text{trans} (B,\hat U) = \left\{ z_{jk} (B, \hat U ):  j,k \in {\mathbb Z}_N \right\} .
\end{align} 

For the 4 unitary operators in Table \ref{tab:props}, we compute the set of transition probabilities 
$S_\text{trans}$ and show the values in the form of images in Figure \ref{fig:tt}.  
From left to right in Figure \ref{fig:tt}a, the panels show the sets $S_\text{trans} (B_{\hat \phi},\hat U)$ for 
the operators $\hat U_{Ta}, \hat U_{Tb}, \hat U_\text{Haar} $ and $\hat U_\text{Drift}$ in the $B_{\hat \phi}$ basis. Figure \ref{fig:tt}b is similar but for the sets $S_\text{trans} (B_{\hat h_0},\hat U) $ computed in 
a basis of eigenvectors of $\hat h_0$.  In Figure \ref{fig:tt}a the basis vectors $\ket{j}$ 
used to compute the transition probability are in order of increasing index $j$ along 
 the horizontal axis  and in order of  increasing index $k$ along the vertical axis.  
In Figure \ref{fig:tt}b is similar but eigenvectors are arranged in order of  increasing eigenvalue of $\hat h_0$.  

We compare the distribution of transition probabilities in the set $S_\text{trans}(B,\hat U)$ to that generated 
from a random state $\ket{\psi}$ 
that is chosen from a uniform distribution in a Hilbert space of the same dimension $N$, 
and consistent with the Haar measure $\mu_\text{Haar}$.  Given $\ket{\psi}$, 
the probability of observing a reference state $\ket{\xi_0}$  is 
\begin{align}
z= |\bra{\xi_0}\ket{\psi} |^2. 
\end{align}
The probability $z$,  
when averaged over the distribution of possible uniformly distributed states $\ket{\psi}$, is 
\begin{align}
p_\text{PT}(z) \approx N e^{-Nz} \label{eqn:PT}.
\end{align}
This distribution is known as the Porter-Thomas distribution \citep{Porter_1956,Brody_1981,Boixo_2018}.
The associated cumulative probability distribution 
\begin{align}
F_\text{PT} (z) = \int_0^z p_\text{PT}(z') dz' \approx 1 - e^{-Nz}.  \label{eqn:CPT}
\end{align}

In Figure \ref{fig:PT} we show the cumulative histograms of transition probabilities 
computed from the 4 different propagators and using two different bases.  
Each cumulative histogram is constructed from a single set of transition probabilities using a 
single basis and a single propagator. In Figure \ref{fig:PT}a we show 
the set $S_\text{trans} (B_{\hat \phi},\hat U_{Ta})$ where $\hat U_{Ta}$ is listed in Table 
\ref{tab:props} and  has  parameters ensuring that the ergodic region covers phase space. 
The Porter-Thomas distribution is computed for a distribution of randomly selected states, but
here we use all pairs of states in a particular basis to compute a set of transition probabilities
caused by a single propagator. 
Figure \ref{fig:PT}a also shows a cumulative distribution generated from the set  
$S_\text{trans}(B_{\hat \phi}, \hat U_\text{Haar})$ where $U_\text{Haar}$ is a single randomly selected unitary. 
Cumulative distributions of $S_\text{trans}(B_{\hat \phi},\hat U_T)$ 
and $S_\text{trans} (B_{\hat \phi},\hat  U_\text{Haar})$ both resemble the cumulative probability distribution 
of the Porter-Thomas distribution (equation \ref{eqn:CPT}), 
which is plotted as a black line in Figure \ref{fig:PT}. 

We also compute sets of transition probabilities in the 
basis $B_{\hat h_0}$ generated from the eigenstates of $\hat h_0$, the unperturbed Hamiltonian 
with the same parameters as used to generate the Floquet propagator $\hat U_{Ta}$. 
In Figure \ref{fig:PT}b we show the cumulative distributions of $S_\text{trans}(B_{\hat h_0}, \hat U_{Ta})$
and $S_\text{trans}(B_{\hat h_0}, \hat U_\text{Haar})$. 
The chaotic operator $\hat U_{Ta}$ is fairly well distributed or delocalized with respect 
to the $B_{\hat h_0}$ basis as well. 
In Figure \ref{fig:PT}c we show a set of transition probabilities 
$S_\text{trans}(B_{\hat \phi},\hat U_{Tb})$ for a Floquet propagator with reduced perturbation strengths.  
This system has a smaller ergodic region that does not cover phase space. 
The cumulative distribution of transition probabilities deviates from that of the Haar-random matrix. 
In Figure \ref{fig:PT}c we show a set of transition probabilities 
$S_\text{trans}(B_{\hat \phi},\hat U_\text{Drift})$ for the drifting propagator with parameters listed in Table \ref{tab:props}. 
The cumulative histogram for this propagator is closer to the Haar-random curve than 
that shown in Figure \ref{fig:PT}a for $\hat U_{Ta}$. 

In Figure \ref{fig:PT} it is difficult to tell the difference between 
a cumulative histograms of transition probabilities and the exponential curve that is typical of the Porter-Thomas distribution.  Slight differences in the distributions are easier to see from the moments 
of the transition probabilities.  
We compute second and third moments from the transition probabilities  
\begin{align} 
m_\text{trans}^{(2)} &= \frac{1}{N^2} \sum_{j,k=0}^{N-1} \left(Nz_{jk}\right)^2 \nonumber \\
m_\text{trans}^{(3)} &= \frac{1}{N^2} \sum_{j,k=0}^{N-1} \left(Nz_{jk}\right)^3    \label{eqn:mms}
\end{align}
for the transition probabilities computed in the two different bases and for 
the 4 different propagators.  The resulting moments are listed 
in Table \ref{tab:mom_trans} along with standard errors that are based on 
the variances divided by the square root of the number of samples (which in this case is $N^2$). 
The first moment is equal to 1 due to normalization so it need not be measured. 
The factor of $N$ is included in the sums in equation \ref{eqn:mms} so that the moments can be 
directly compared to those of an exponential distribution $p(z) = e^{-z}$ 
which has moments $\langle z^k \rangle  = k!$. 
As expected, the cumulative distributions that deviate more strongly from the exponential distribution
 in Figure \ref{fig:PT} have larger moments in Table \ref{tab:mom_trans}. 

\begin{table}[htbp] \centering
\caption{Moments of transition probabilities\footnote{Moments for the transition probabilities are defined 
in equations \ref{eqn:transprobs} and \ref{eqn:mms} and were computed 
for four unitaries with dimension $N=81$.}  \label{tab:mom_trans}}
\begin{ruledtabular}
\begin{tabular}{lllll}
 Unitary &    Basis & $ m_\text{trans}^{(2)} $ & $ m_\text{trans}^{(3)}$ \\
\colrule
$\hat U_{Ta}$      & $B_{\hat \phi} $ & 2.44 $\pm$ 0.09 & 10.13 $\pm$ 0.77 \\
$\hat U_{Ta}$      & $B_{\hat h_0} $ & 2.47 $ \pm$  0.10 & 10.96 $\pm$ 1.22 \\
$\hat U_{Tb}$      & $B_{\hat \phi} $ & 3.48 $\pm$ 0.16 & 21.68 $\pm$  1.87 \\
$\hat U_{Tb}$       & $B_{\hat h_0} $ &11.59 $\pm$ 1.51 &  346 $\pm$ 82  \\ 
$\hat U_\text{Drift}$ & $B_{\hat \phi} $ & 1.97 $\pm $ 0.05 & 5.68 $\pm $   0.27  \\
$\hat U_\text{Drift}$ & $B_{\hat h_0} $ & 2.23 $\pm $ 0.06 & 6.33 $\pm $   0.39  \\
$\hat U_\text{Haar}$ & $B_{\hat \phi} $ & 1.97 $\pm$ 0.05 &  5.56 $\pm$  0.24 \\
$\hat U_\text{Haar}$ & $B_{\hat h_0} $ &1.98 $\pm$ 0.05 & 5.81 $\pm$ 0.31 \\
\colrule
Porter-Thomas\footnote{If the transition probabilities 
obeyed a exponential (Porter-Thomas) distribution, they would have the values 
given in this last row.}    && 2  & 6 \\
\end{tabular}
\end{ruledtabular}
\end{table} 

\subsection{Inverse Participation Ratios}

\begin{figure*}[htbp]\centering
\includegraphics[width=3.0truein]{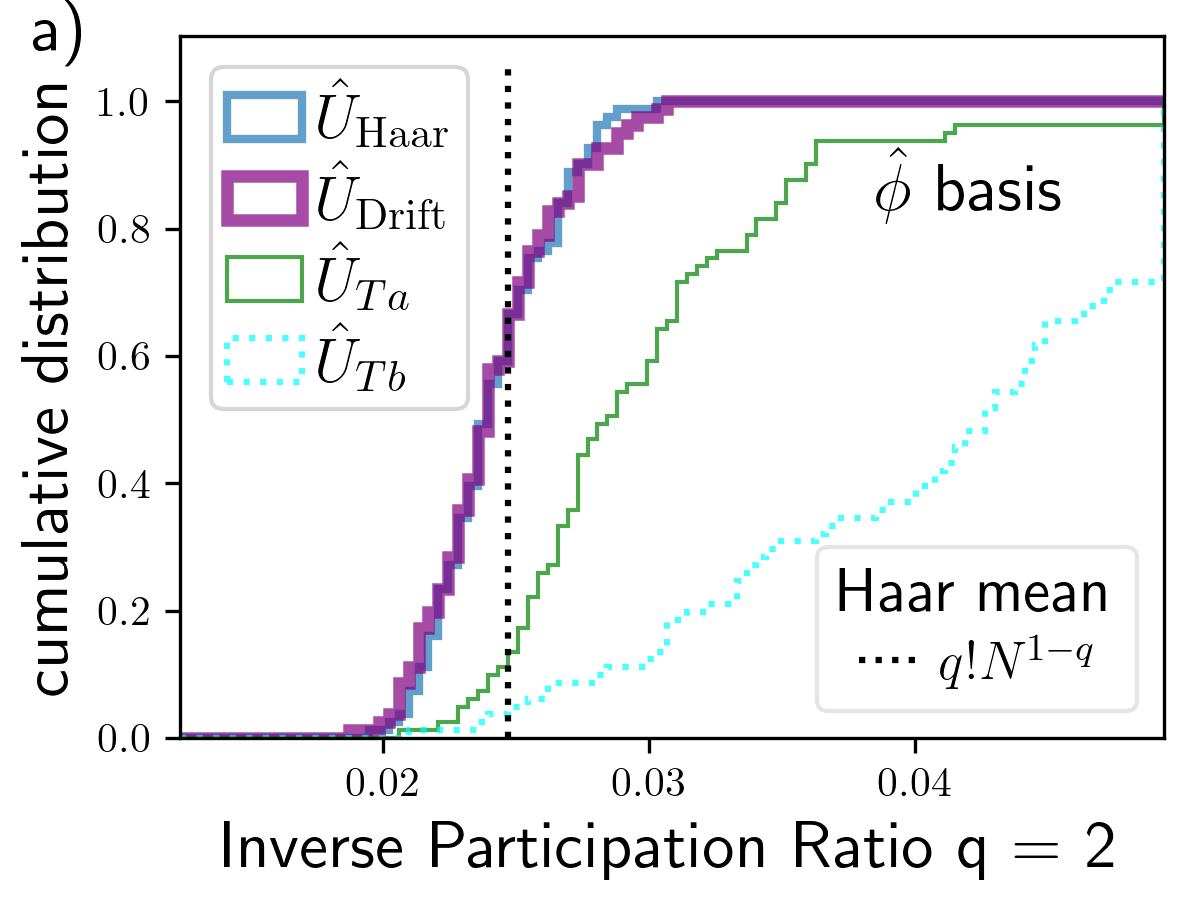}
\includegraphics[width=3.0truein]{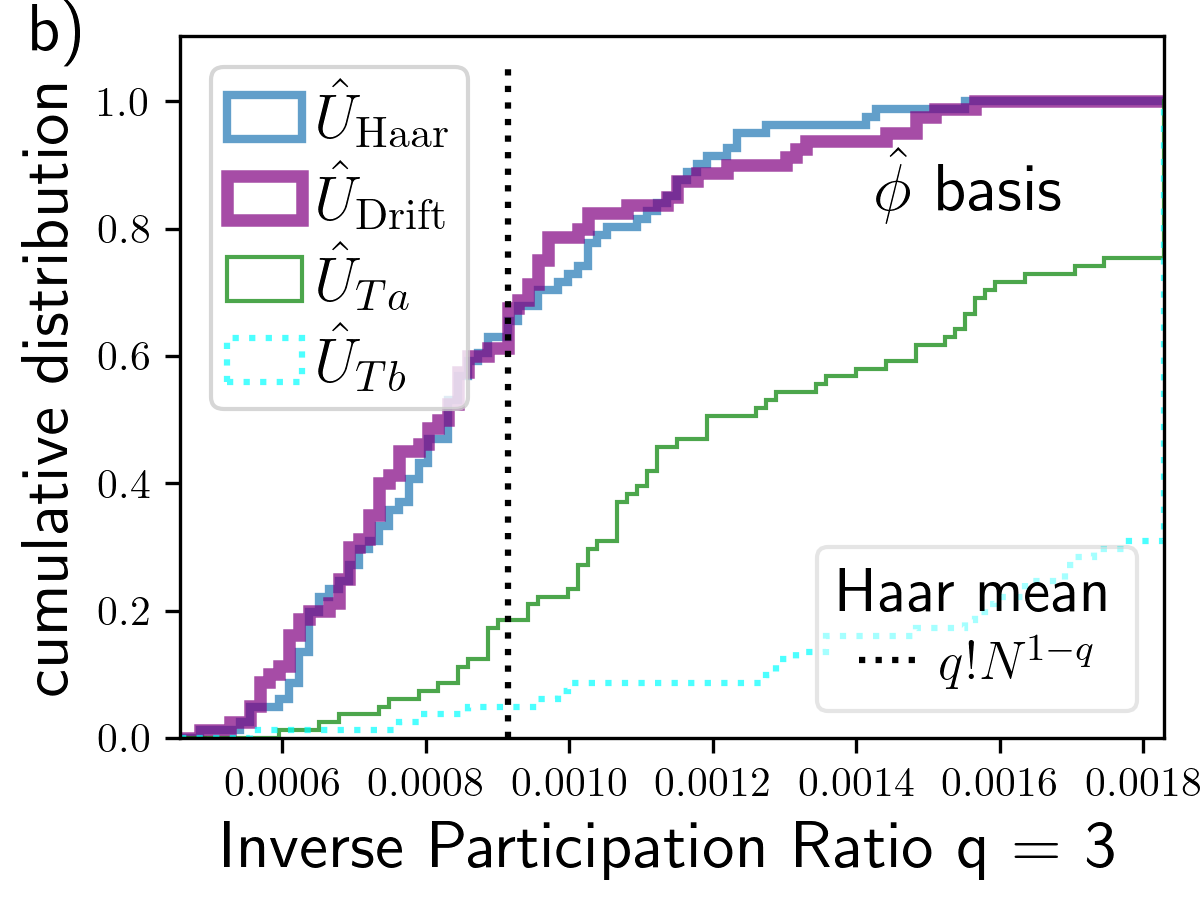}
\includegraphics[width=3.0truein]{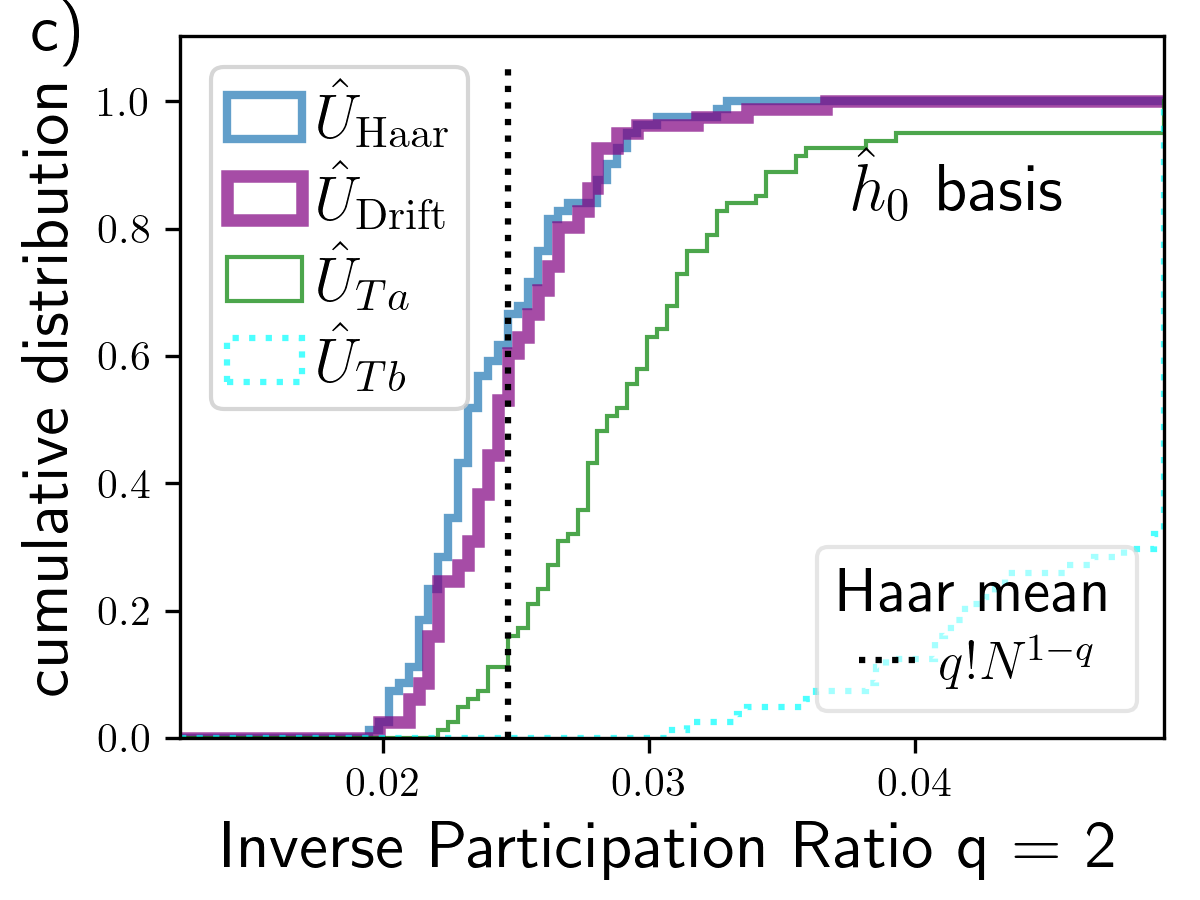}
\includegraphics[width=3.0truein]{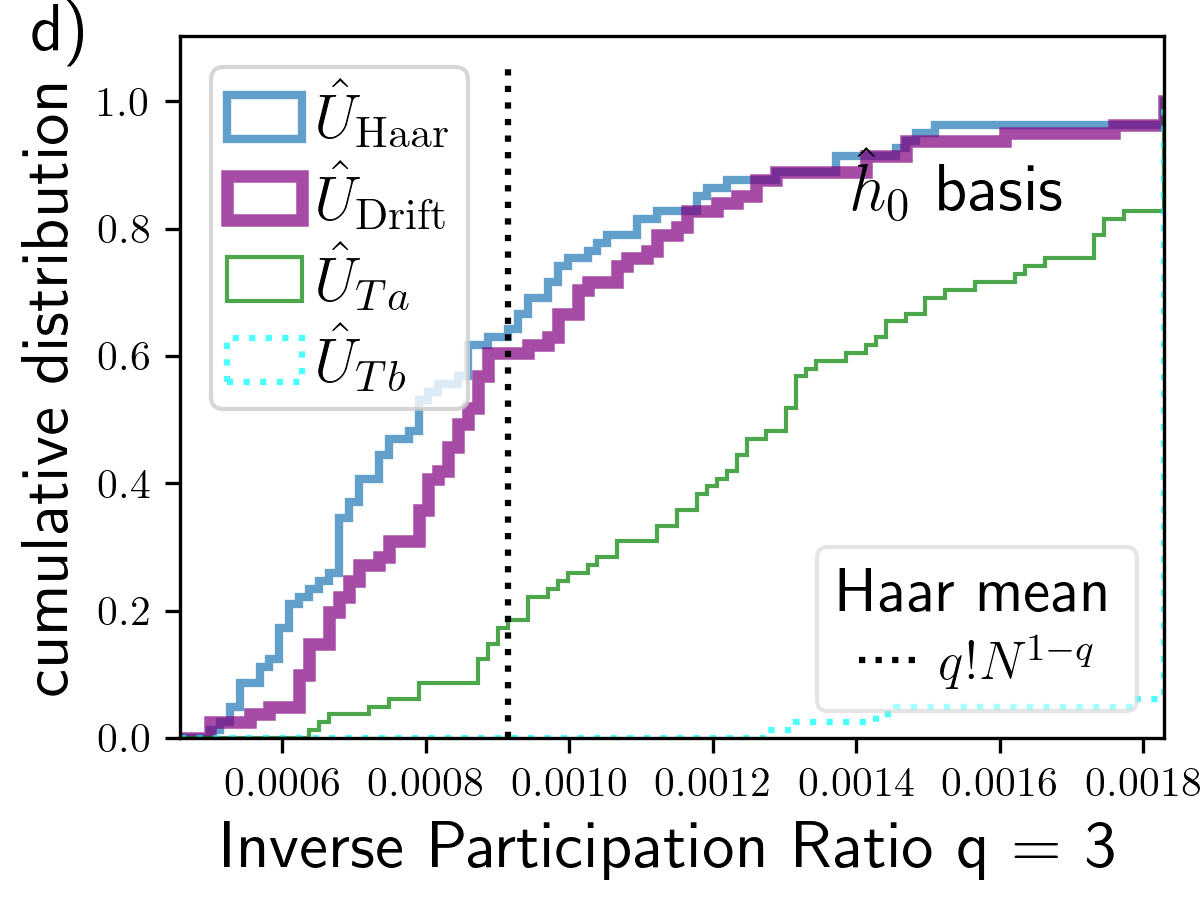}
\caption{a) We plot cumulative histograms of the set of 
inverse participation ratios $I_{q=2}(B_{\hat \phi},\hat U)$. 
These are computed via equation \ref{eqn:IqB} for each 
basis state in the $\ket{j}$ basis (eigenvectors of $\hat \phi$) and for the 4 propagators 
listed in Table \ref{tab:props} with $N=81$.
b) Same as a) except the inverse participation rates have $q=3$.
c) Same as a) except in a basis $B_{\hat h_0}$ composed of eigenvectors of $\hat h_0$. 
d) Same as a) except in basis $B_{\hat h_0}$ and with $q=3$.  
\label{fig:IPR} }
\end{figure*}

A quantity called 
the {\it inverse participation ratio} \citep{Wegner_1980,Liu_2025,Lami_2025,Claeys_2025} 
is also a way to measure localization or the concentration of a state.  

\begin{definition} \label{def:IPR}
With an orthonormal basis $B = \{ \ket{v_j}: j\in  {\mathbb Z}_N  \}$, 
the $q$-th {\bf inverse participation ratio} of a particular quantum state $\ket{\psi}$  is 
\begin{align}
I_q(B, \ket{\psi}) \equiv \sum_{j=0}^{N-1} \left|\bra{v_j}\ket{\psi}\right|^{2q} \text { for } \ket{v_j} \in B
\label{eqn:IqB}
\end{align}
where $q$ is a positive integer.  
\end{definition}

A fully localized state with $\ket{\psi}$ equal to a basis element
 has inverse participation ratio $I_q = 1$, with respect to that basis, whereas 
if the state spans the entire space homogeneously, 
$\ket{\psi} = \sum_{j=0}^{N-1} \frac{1}{\sqrt{N}} \ket{v_j}$, 
the inverse participation ratio is $I_q = N^{1-q}$. 
If we take the expectation value of the participation ratio for a distribution 
of states  $\{ \ket{\psi} \}$ that are uniformly  distributed,  the average 
\begin{align}
 {\mathbb E}_{\ket{\psi}\sim \text{Haar}} \left[ I_{q}(B, \ket{\psi})  \right] = q! N^{1-q} 
 \label{eqn:Haar_IPR}
 \end{align}
\citep{Claeys_2025}.

We compare how well the unitaries in Table \ref{tab:props} cause delocalization by computing 
the inverse participation ratio 
\begin{align}
I_q(B, \hat U \ket{v_j}) 
\end{align}
for unitary $\hat U$ operating on 
 $\ket{v_j}$, one of the basis elements in the basis $B$. 
For unitary $\hat U$ we look at the set of non-negative real numbers  
\begin{align}
S_{q}(B,\hat U) = \left\{  I_q(B, \hat U \ket{v_j}):  \text{ for } j  \in {\mathbb Z}_N   \right\}
. \label{eqn:Set_IPR}
\end{align}

In Figures \ref{fig:IPR} we show cumulative histograms of inverse participation ratios computed for
the 4 unitaries listed in Table \ref{tab:props}.   In Figure \ref{fig:IPR}a we show  
$S_{q=2}( B_{\hat \phi}, \hat U) $  with each plotted line showing  
 one of the 4 unitaries. Figure \ref{fig:IPR}a is similar but with $q=3$. 
 Figure \ref{fig:IPR}c,d  are similar to Figure \ref{fig:IPR}a,b but use the basis $B_{\hat h_0}$. 
The dotted vertical lines show the expectation or mean value  (of equation \ref{eqn:Haar_IPR}) 
expected for a Haar-random  set of unitaries. 
Figure \ref{fig:IPR} illustrates that the Haar-random operator $\hat U_\text{Ha}$ and 
the drifted propagator $\hat U_\text{Drift}$ have similar distributions of
inverse participation ratios, whereas those of the Floquet propagators $\hat U_{Ta}$ and 
$\hat U_{Tb}$ are significantly lower, with the hybrid $\hat U_{Tb}$ 
having the lowest inverse participations ratios.  

\subsection{Short summary}

Figure \ref{fig:SSab} illustrates that a strongly 
perturbed version of periodically perturbed Harper model can be chosen so that the 
 classical Hamiltonian has a wide ergodic region that covers
phase space and lacks resonant islands.   
Husimi distributions generated from quantum operators with the same parameters 
resemble the classical orbits (as shown in  Figure \ref{fig:Hus}). 

We compared properties of 4 unitary operators, two of which are Floquet propagators,
one that is drawn from a Haar random distribution, and a drifted version 
of the periodically perturbed system.  
One of the Floquet propagators was chosen to be fully ergodic and the other 
is a hybrid with ergodic and non-ergodic regions in phase space.  

Notions of quantum chaos include the idea that a quantum propagator is spread over phase space 
\citep{Feingold_1985} or that position in momentum observables appear to be spread out
 in the eigenbasis of the propagator \citep{Peres_1984}. 
The operators that we describe as ergodic share a number of properties with the Haar-random
unitary:  Husimi distributions computed from their eigenstates are evenly distributed across phase 
space (Figure \ref{fig:Hus}),  transition probabilities computed in both $\hat \phi$ (angle) and $\hat h_0$ (unperturbed energy) bases appear randomly distributed and their distribution resembles a Porter-Thomas distribution (Figures \ref{fig:tt} and \ref{fig:PT}).   Eigenvalue spacing distributions of ergodic
operators are similar to the circular unitary random matrix ensemble (Figure \ref{fig:spacing}).  
Inverse participation ratios measured in both $\hat \phi$ and $\hat h_0$ bases illustrate that the
drifting system and ergodic Floquet systems are more mixed or anti-concentrated 
in comparison to the propagator containing resonant islands (Figure \ref{fig:IPR}). 
In appendix \ref{ap:op_coefs} 
 we show that for the ergodic unitaries the distribution of amplitudes in the operator basis 
constructed from displacement operators is also randomly distributed,  
and has distribution consistent with a Porter-Thomas distribution (see Figures \ref{fig:Wvals} and \ref{fig:wcumu}).
We have illustrated in different ways that an ergodic propagator exhibits properties similar to
those of a Haar-random matrix. 

Uniform mixing or anti-concentration could be related to 
the sensitivity of an operator to the parameters that are used to generate it.  
Evidence of uniform mixing would imply that that insensitive subspaces are not present. 
Hence,  we suspect that operators we chose for their ergodic properties are 
more suitable for generating quantum samplers. 

\section{Constructing quantum samplers}
\label{sec:samplers}

Given a distribution of unitary operators, it is desirable to determine if 
the distribution has moments similar or equal to a Haar-random distribution (computed with the 
Haar measure).   If the k-th moments are equal to that of a Haar-random distribution 
then the distribution is called a {\bf k-design}. 
One way to characterize a distribution of unitaries is with quantities called the {\it k-frame potentials} 
\citep{Gross_2007,Scott_2008,Mele_2024}. 

\begin{definition} \label{def:Framep}
Given a distribution $\nu$ of unitaries in U$(N)$  the {\bf k-frame potential }
\begin{align}
{\cal F}_\nu^{(k)}  \equiv  {\mathbb E}_{\hat U,\hat V \in \nu}  \left[  \left|\tr \hat U\hat V^\dagger\right|^{2k} \right]
 \label{eqn:Frame}
.
\end{align}
Here $k$ is a positive integer.  
\end{definition}

A distribution of unitaries in U$(N)$ is a {\it k-design} if and only if its k-frame potential is 
equal to that of a Haar-random distribution;  
\begin{align}
{\cal F}_\nu^{(k) }= {\cal F}_{\mu_\text{Haar}} = k!
\end{align}
\citep{Mele_2024}, where $\mu_\text{Haar}$ represents the Haar measure. 
The minimum possible value of the k-frame potential is equal to that of the Haar measure \citep{Scott_2008}. 

It is challenging to derive analytical estimates for k-frame potentials \citep{Brandao_2016,Nakata_2017}, 
hence we estimate them numerically. 
We numerically generate samples of pairs of unitaries $\hat U_i, \hat V_i$  from a distribution $\nu$ 
of unitaries 
with index $i$ referring to an individual sample. From each pair of randomly generated samples 
we compute 
\begin{align}
z_i = \left|\tr \hat U_i \hat V_i^\dagger \right|^2  . \label{eqn:zzz}
\end{align}
We estimate the k-frame potential of the distribution from the sample's average  
\begin{align}
{\cal F}_\nu^{(k)} \approx  \langle z^k \rangle  = \frac{1}{N_U}\sum_{i=1}^{N_U}  z_i^k 
\label{eqn:Fp_num}
\end{align}
where $N_U$ is the number of pairs of unitaries generated. 
To estimate the accuracy of the estimate we use the standard error 
\begin{align}
\sigma_k = \frac{1}{\sqrt{N_U}} \sqrt{\langle z^{2k}\rangle  - (\langle z^k \rangle)^2}. \label{eqn:std_err}
\end{align}
This error gives 
 an estimate for the scatter in the generated values of the $z$ values in equation \ref{eqn:zzz}
or their powers. 

\subsection{Floquet Samplers}

In this section we explore generating samples of unitaries from the periodic 
Floquet system with propagator computed over a period of $2\pi$ 
from the Hamiltonian operator defined by equations \ref{eqn:hath}, \ref{eqn:hath0} and \ref{eqn:hath1}. 
To describe a sampler, we give distributions for the parameters of 
the Hamiltonian ($a, b, \epsilon, \mu, \mu', \phi_0, \tau_0$) from which the propagator is derived
(equation \ref{eqn:UT} with $T = 2 \pi$).   
A set of randomly generated unitaries is generated by 
first selecting a set of parameters from their distributions and then computing 
the Floquet propagator with these selected parameters. 
A list of Floquet samplers is given in Table \ref{tab:fl_samplers}.  
The samplers are labelled with a letter D if angles ($b, \phi_0$ or $\tau_0$) are chosen from 
a discrete set of values, otherwise they are described with the letter C, implying 
that parameters are chosen from uniform distributions within particular maximum
and minimum  values. 
Each sampler label lists the parameters that are drawn from distributions.  Remaining 
parameters are fixed and have fiducial values listed at the top of the table.  
For most of these samplers, parameters are similar to those of $\hat U_{Ta}$ 
(see Table \ref{tab:props} and discussed in section \ref{sec:comp4}) which is fully ergodic, with the exception 
of that labelled C$_w b\tau\mu\mu'$ which has smaller perturbation strength parameters
and resembles that of $\hat U_{Tb}$ which is hybrid, containing both ergodic and integrable regions 
in phase space. 

We randomly choose parameters from their distributions and generate a 
sample of $N_U$ pairs of unitary propagators.  
From these pairs we estimate the $k = 1,2,3$ frame potentials (via equation 
\ref{eqn:Fp_num}) along with standard errors computed with equation \ref{eqn:std_err}. 
Numerically estimated  k-frame potentials are listed in Table \ref{tab:frames} with the number of samples
used. 
We also list the theoretical values for a Haar distribution and a control sample of values 
measured numerically from a generated set of Haar-random unitaries 
of the same dimension and with a similar number of samples as used to estimate 
 other k-frame potentials. 

\begin{figure}[htbp]\centering
\includegraphics[width=3.3truein,trim = 8 0 0 0,clip]{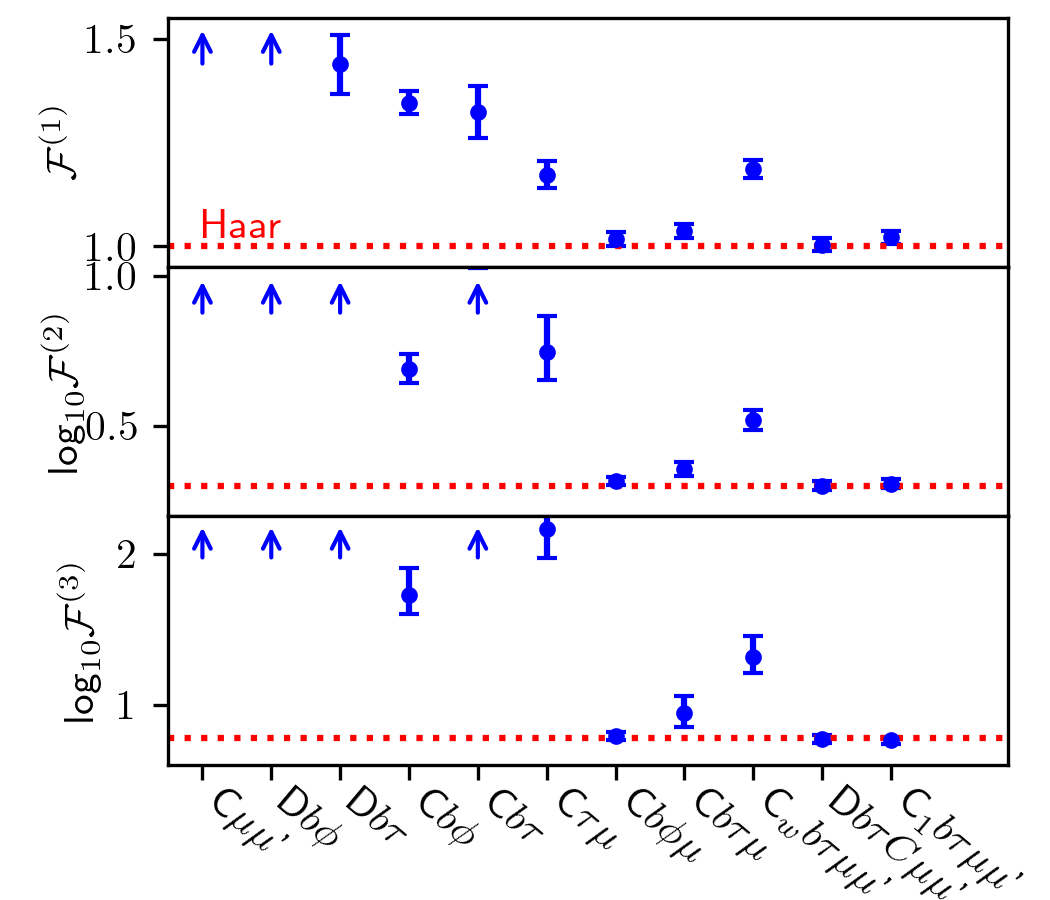}
\caption{Numerically estimated k-frame potentials of Floquet samplers 
generated by constructing propagators from distributions of parameters that describe them.  
From top to bottom we show k-frame potentials with $k=1$, 2 and 3, respectively. 
The vertical axes are on a log scale for $k=2,3$ but on a linear scale for $k=1$.  
The unitary distributions were generated using sample distributions listed in Table \ref{tab:fl_samplers}. 
The parameters were chosen to be close to those of the operator $\hat U_{Ta}$ (discussed
in section \ref{sec:comp4}) which is fully ergodic 
except for the sample on the right which was chosen to be closer to the operator $\hat U_{Tb}$
which contains resonant islands.   
Frame potential values and the number
of samples used in estimating them are listed in Table \ref{tab:frames}. 
The red dotted line shows the $k!$ values 
expected for a Haar-random sample. 
\label{fig:fp_fl}}
\end{figure} 

\begin{table}[!htb]
\caption{Description of Floquet Samplers\footnote{The $k=1,2,3$ frame potentials calculated from generated pairs of unitaries using these samplers are shown in Figure \ref{fig:fp_fl} and listed in Table \ref{tab:frames}.} \label{tab:fl_samplers}}
\begin{ruledtabular}
\begin{tabular}{llll}
& Parameters of $\hat h$\\
\colrule
Fiducial\footnote{Parameters of the Hamiltonian operator (specified by equations 
\ref{eqn:hath}, \ref{eqn:hath0}, \ref{eqn:hath1}) 
used to generate a propagator (via equation \ref{eqn:UT} with duration $T=2\pi$) 
are those in the list denoted `fiducial' unless they are drawn from a set described in the second half of the table.} & $N=51,a=3,b=\phi_0=\tau_0= 0$ \\
& $\epsilon=3,\mu=3,\mu'=3.1$ \\
\colrule
C\footnote{The prefix C implies that parameters were
drawn from uniform distributions with low and high values given on the right. 
}$\mu\mu'$ & $\mu \in [3,7),\mu'\in[3.1,7.1)$ \\
D\footnote{The prefix D implies that angles were 
drawn uniformly from a discrete set of values.}$b\phi$ & $b,\phi_0 \in \left\{ \frac{2 \pi j}{N} \text{ integer } j \right \}$ \\
D$b\tau$ & $b,\tau_0 \in \left\{ \frac{2 \pi j}{N} \text{ integer }  j \right\}$  \\
C$b\phi$ & $b,\phi_0 \in [0, 2 \pi ) $ \\
C$b\tau$ & $b,\tau_0 \in [0, 2 \pi ) $ \\
C$\tau\mu$ & $\tau_0 \in [0, 2 \pi ), \mu \in [3,6) $ \\
C$b\phi\mu$ & $b,\phi_0 \in [0, 2 \pi ), \mu \in [3,6) $ \\
C$b\tau\mu$ & $b,\tau_0 \in [0, 2 \pi ), \mu \in [3,6) $ \\
C$_w b\tau\mu\mu'$ & $b,\tau_0 \in [0, 2 \pi ), \mu \in [1,4),\mu'\in[0,2) $ \\
D$b\tau$C$\mu\mu'$  & $b,\tau_0 \in \left\{ \frac{2 \pi j}{N} \text{ integer }  j \right\}$, $ \mu \in [3,6),\mu'\in[3.1,5.1) $ \\
C$_1 b\tau\mu\mu'$\footnote{C$_1b\tau\mu\mu'$ through C$_5b\tau\mu\mu'$ are the same except for dimension $N$.}  & $b,\tau_0 \in [0, 2 \pi ), \mu \in [3,6),\mu'\in[3.1,5.1),N=51 $ \\
C$_2 b\tau\mu\mu'$ & $b,\tau_0 \in [0, 2 \pi ), \mu \in [3,6),\mu'\in[3.1,5.1),N=30 $ \\
C$_3 b\tau\mu\mu'$ & $b,\tau_0 \in [0, 2 \pi ), \mu \in [3,6),\mu'\in[3.1,5.1),N=70 $ \\
C$_4 b\tau\mu\mu'$ & $b,\tau_0 \in [0, 2 \pi ), \mu \in [3,6),\mu'\in[3.1,5.1),N=100 $ \\
C$_5 b\tau\mu\mu'$ & $b,\tau_0 \in [0, 2 \pi ), \mu \in [3,6),\mu'\in[3.1,5.1),N=150 $ \\
\end{tabular}
\end{ruledtabular}
\end{table}

\begin{table}[!htb]
\caption{k-Frame Potentials of Samplers\footnote{The numerically estimated k-frame potentials for $k=1,2,3$ of the samplers listed in Table \ref{tab:fl_samplers} and \ref{tab:dr_samplers}.} \label{tab:frames}}
\begin{ruledtabular}
\begin{tabular}{llll}
Sampler\footnote{The Floquet samplers are those with labels beginning with C or D and  are computed using $N_U=4000$ pairs of randomly generated propagators. The Drifted samplers are those with labels beginning with Dr and are computed using $N_U=2000$ pairs of randomly generated propagators.} & ${\cal F}^{(1)}$   & ${\cal F}^{(2)}$  & ${\cal F}^{(3)}$ \\
\colrule
C$\mu\mu'$  & 4.22 $\pm$ 0.72 &  2118 $\pm$ 1314 &  $(36.6  \pm 2.6)\! \times\! 10^6$ \\ 
D$b\phi$  & 3.33 $\pm$ 1.13 &  5112 $\pm$ 2928 &  $(13.2  \pm 7.6)\! \times\! 10^6$ \\ 
D$b\tau$  & 1.44 $\pm$ 0.07 &  22.31 $\pm$ 10.26 &  2662 $\pm$ 1726 \\ 
C$b\phi$  & 1.35 $\pm$ 0.03 &  4.88 $\pm$ 0.55 &  53.13 $\pm$ 17.82 \\ 
C$b\tau$  & 1.32 $\pm$ 0.06 &  17.71 $\pm$ 11.85 &  2703 $\pm$ 2555 \\ 
C$\tau\mu$  & 1.17 $\pm$ 0.03 &  5.58 $\pm$ 1.34 &  145.4 $\pm$ 81.9 \\ 
C$b\phi\mu$  & 1.02 $\pm$ 0.02 &  2.07 $\pm$ 0.07 &  6.20 $\pm$ 0.36 \\ 
C$b\tau\mu$  & 1.04 $\pm$ 0.02 &  2.27 $\pm$ 0.12 &  8.78 $\pm$ 2.03 \\ 
C$_wb\tau\mu$$\mu$'  & 1.19 $\pm$ 0.02 &  3.30 $\pm$ 0.26 &  20.63 $\pm$ 5.61 \\ 
D$b\tau$C$\mu\mu'$      & 1.00 $\pm$ 0.02 &  2.00 $\pm$ 0.07 &  5.89 $\pm$ 0.38 \\ 
C$_1b\tau\mu$$\mu$'  & 1.02 $\pm$ 0.02 &  2.03 $\pm$ 0.06 &  5.81 $\pm$ 0.30 \\ 
C$_2b\tau\mu$$\mu$'  & 1.00 $\pm$ 0.02 &  2.12 $\pm$ 0.09 &  7.36 $\pm$ 0.78 \\ 
C$_3b\tau\mu$$\mu$'  & 1.01 $\pm$ 0.02 &  2.14 $\pm$ 0.11 &  7.72 $\pm$ 1.39 \\ 
C$_4b\tau\mu$$\mu$'  & 1.02 $\pm$ 0.02 &  2.10 $\pm$ 0.07 &  6.41 $\pm$ 0.41 \\ 
C$_5b\tau\mu$$\mu$'  & 1.00 $\pm$ 0.02 &  2.03 $\pm$ 0.07 &  6.25 $\pm$ 0.44 \\ 
\colrule
Dr$_1\dot b\dot \mu$  & 1.26 $\pm$ 0.08 &  15.9 $\pm$ 10.4 &  1613 $\pm$ 1482 \\ 
Dr$_2b \dot \mu$       & 1.02 $\pm$ 0.02 &  2.07 $\pm$ 0.10 &  6.10 $\pm$ 0.46 \\ 
Dr$_3\tau\dot \mu$    & 1.00 $\pm$ 0.02 &  2.02 $\pm$ 0.11 &  6.53 $\pm$ 0.76 \\ 
Dr$_4\tau\dot \mu$    & 1.03 $\pm$ 0.02 &  2.26 $\pm$ 0.16 &  8.86 $\pm$ 1.84 \\ 
Dr$_5\tau\dot \mu\dot \mu'$  & 1.00 $\pm$ 0.02 &  1.96 $\pm$ 0.09 &  5.71 $\pm$ 0.49 \\ 
\colrule
Haar-control\footnote{This row lists values measured numerically with $N_U=4000$ pairs of Haar-uniform generated unitaries. }& 1.00 $\pm$ 0.02 &  1.99 $\pm$ 0.07 &  5.79 $\pm$ 0.38 \\ 
Haar-ideal\footnote{This row lists theoretical values ($k!$) of a Haar-uniform distribution. } & 1  & 2  & 6 \\
\end{tabular}
\end{ruledtabular}
\end{table}

The numerically estimated k-frame potentials (with values listed in Table \ref{tab:frames}) 
for the Floquet samplers with $N=51$ are plotted in Figure \ref{fig:fp_fl}. 
Comparison of the k-frame potentials computed from discrete uniform distributions of two 
angles $b, \phi_0 $ or $b, \tau_0$ are higher than those computed from 
continuous uniform distributions of these angles.   
The sampler $C\mu\mu'$ which only varies perturbation strengths is one of the worst samplers 
(with the highest k-frame potentials).   Samplers 
have lower frame potentials if they are based on distributions of both angles and perturbation strengths. 
With more three or more varied parameters,  
the numerically estimated estimates for the k-frame potentials are similar 
to those of the Haar-uniform distribution, with the exception of 
the sampler labelled C$_wb\tau\mu\mu' $ which has lower strength 
perturbations compared to the similar sampler labelled  C$_1b\tau\mu\mu' $.  
The sampler C$_wb\tau\mu\mu' $, because it has lower perturbation strengths  
generates unitaries similar to $\hat U_{Tb}$, which has  both ergodic and integrable 
regions in phase space.  The lower value of the k-frame potentials for the sampler C$_1b\tau\mu\mu' $, in comparison to the sampler C$_wb\tau\mu\mu'$ 
suggests that a fully and strongly ergodic underlying operator facilitates generating a more uniform sampler. 

\subsection{Samplers from drifted systems}

\begin{figure}[htbp]\centering
\includegraphics[width=3.2truein]{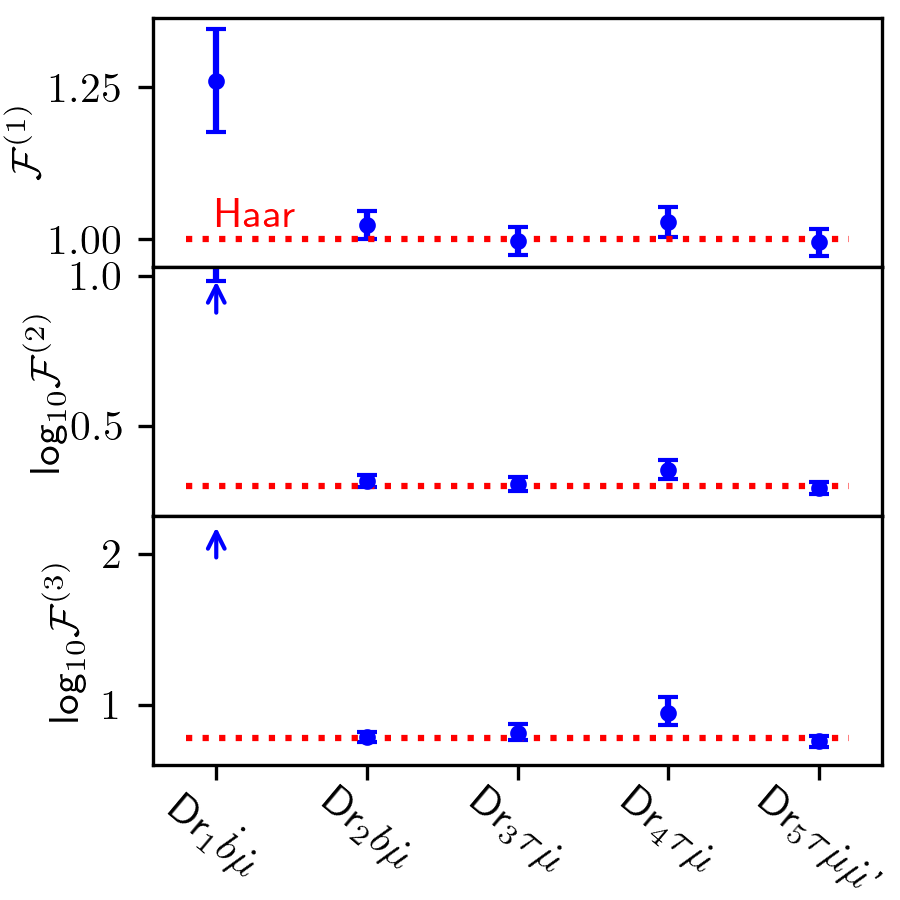}
\caption{Numerically estimated k-frame potentials of samples based on drifted Floquet systems. 
The unitary distributions were generated using sample distributions listed in Table \ref{tab:dr_samplers}. 
The parameters were chosen to be close to those of the operator $\hat U_\text{Drift}$ discussed
 in section \ref{sec:comp4}.
From top to bottom we show k-frame potentials 
with $k=1$, 2 and 3, respectively.  Frame potential values and the number
of samples used in estimating them are listed in Table \ref{tab:frames}. 
The vertical axes are on a log scale for $k=2,3$ but on a linear scale for $k=1$.  
The red dotted line shows the $k!$ values expected for a Haar-random sample. 
\label{fig:fp_var}}
\end{figure} 

\begin{table}[!htb]
\caption{Description of Drifted Samplers\footnote{The $k=1,2,3$ frame potentials calculated from 2000 generated pairs of unitaries are shown in Figure \ref{fig:fp_var} and listed in Table \ref{tab:frames}. } \label{tab:dr_samplers}}
\begin{ruledtabular}
\begin{tabular}{llll}
& Parameters of $\hat h$\\
\colrule
Fiducial\footnote{Parameters of the drifted Hamiltonian operator 
are those in the list denoted `fiducial' unless they are drawn from a set described in 
the samplers listed in the second half of the table.} & $N=51,a=3,b_0\footnote{Initial values of the parameters $b,\mu, \mu'$ are denoted $b_0, \mu_0, \mu_0'$.}=\phi_0=\tau_0= 0$ \\
& $\epsilon=3.01,\mu_0=1,\mu'_0=0.5, n_\text{periods}=3$ \\
& $\dot b = 0, \dot \mu = 0.1, \dot \mu' = 0.15 $\footnote{Drift rates in  parameters $b,\mu, \mu'$  are denoted $\dot b, \dot \mu$ and $\dot \mu'$.  Time derivatives of parameters with respect to $\tau$ and are zero unless specified. } \\
\colrule
Dr$_1 \dot b \dot \mu$ & $\dot b\in[0,0.1), \dot \mu \in [0.1,0.6)$ \\  
Dr$_2b \dot \mu $        & $b\in[0,2\pi), \dot \mu \in [0.1,0.6)$ \\ 
Dr$_3 \tau \dot \mu$    &  $\tau_0\in[0,2\pi), \dot \mu \in [0.1,0.6)$\\ 
Dr$_4 \tau \dot \mu $   & $\dot b =0.1$,  $\tau_0\in[0,2\pi), \dot \mu \in [0.1,0.6)$\\ 
Dr$_5 \tau \dot \mu \dot \mu'$ & $\tau_0\in[0,2\pi), \dot \mu \in [0.1,0.6), \dot \mu' \in [0.15,0.65)$ \\ %
\end{tabular}
\end{ruledtabular}
\end{table}

In Table \ref{tab:dr_samplers} we list samplers constructed from drifting systems. 
These contain more free parameters to describe than  the Floquet samplers as 
in addition to the same parameters, we can also choose distributions of drift rates 
for each parameter.  We find that 
 distributions of only 2 or 3  parameters is sufficient to achieve an approximate 3-design.   
 To achieve the same level of approximation 
to a Haar-random  distribution, 
the number of parameters that we need to choose with  
distributions for the drifting samplers seems to be 1 fewer than for the Floquet samplers.   
We found in section \ref{sec:comp4} that the drifted propagator was closer in many respects 
to a random matrix than the Floquet propagators, and this could have aided in giving increased
sensitivity to randomly chosen parameters in the drifted systems. 
When systems are slowly varied, there would be numerous close approaches (also called
avoided crossings) in energy levels \citep{Wilkinson_1989,Zakrzewski_2023} 
of the effective Floquet Hamiltonian, 
which could have effectively enhanced the sensitivity of the drifted samplers to their randomly selected 
parameters. 


\subsection{How good are the approximate k-designs?}

Our best samplers have k-frame potentials similar to those of a Haar-random distribution and 
so can be called frame potential $\varepsilon$-approximate k-designs. 

We adopt definition 43 by \citet{Mele_2024}. 
\begin{definition} 
The distribution $\nu$ of unitary operators is 
a {\bf frame potential $\varepsilon$-approximate k-design} if and only if
\begin{align}
\sqrt{ { \cal F}_\nu ^{(k)} - { \cal F}_{\mu_\text{Haar}}^{(k)}}  \le \varepsilon  \label{eqn:varepsilon}
\end{align}
The argument inside the square root is always non-negative. 
\end{definition}

In Table \ref{tab:frames}, the better k-potential frame designs (such as sampler C$_1b\tau \mu \mu'$ 
and Dr$_4\tau \dot \mu \dot \mu'$) 
 have  difference  $|{ \cal F}_\nu ^{(k)} - { \cal F}_{\mu_\text{Haar}}^{(k)} | < \sigma_k$ 
where $\sigma_k$ is the standard error estimated from
the variance of values computed in the sequence of traces (equation \ref{eqn:std_err}). 
We roughly estimate $\varepsilon$ from the square root of the standard error for the 
C$_1b\tau \mu \mu'$ sampler, which is one of the best in the Floquet group, 
and Dr$_4\tau \dot \mu \dot \mu'$ which is the best sampler in the drifted group, 
and list these values in Table \ref{tab:varepsilon} for $k=1,2$ and 3. 
The $\varepsilon$ values are similar to those given by the standard errors estimated for the Haar control 
sample generated in the same dimension $N$ and with the same number of samples as for the Floquet sampler.  

\begin{table}[htbp]\centering
\caption{$\varepsilon$ estimates\footnote{$\varepsilon$ is estimated from the square root of the standard errors for  two samplers with properties given in Tables \ref{tab:fl_samplers} and \ref{tab:dr_samplers} and
numerical estimates listed in Table \ref{tab:frames}.  Their size is primarily dependent 
upon the number of samples generated.  }  for frame potential $\varepsilon$-approximate k-designs}
\label{tab:varepsilon}
\begin{ruledtabular}
\begin{tabular}{llll}
Sampler                & $k=1$ & $k=2$ & $k=3$\\
\colrule
C$_1b\tau\mu\mu'$              & 0.13 & 0.25 & 0.55 \\
Dr$_5\tau\dot \mu\dot \mu'$ & 0.15 & 0.31 & 0.70 \\
\end{tabular}
\end{ruledtabular}
\end{table}

Figure \ref{fig:cufps} shows the cumulative distributions for operators 
drawn from our best two samplers, C$_1b\tau\mu\mu'$   and  Dr$_5\tau\dot \mu\dot \mu'$, 
and illustrates that the distribution is quite close to the exponential distribution expected for 
a Haar-random sample. 

\begin{figure}[htbp]\centering
\includegraphics[width=3.3truein,trim = 0 0 0 0,clip]{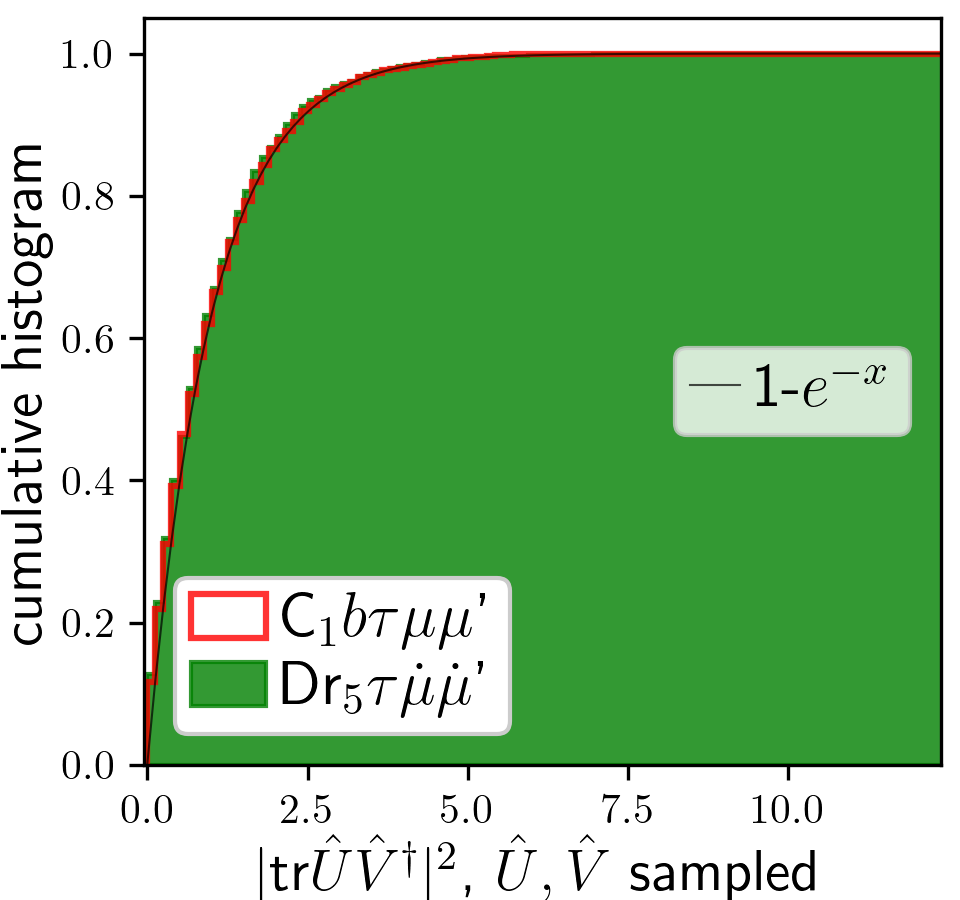}
\caption{
We show the cumulative distribution of $|\tr U V^\dagger|^2$ values for $\hat U, \hat V$ 
drawn randomly from our two best samplers, C$_1b\tau\mu\mu'$ and  Dr$_5\tau\dot \mu\dot \mu'$.
The thin black  line shows the exponential distribution values 
expected for a Haar-random sample. 
\label{fig:cufps}}
\end{figure} 

\subsection{Parameters that affect the unitary propagator by conjugation} 
\label{sec:conj}

Our best samplers are based on distributions of angles and perturbation strengths.
We discuss the role of these two types of parameters separately. 

Conjugation of an operator $\hat A$ by operator unitary $\hat U$ is the map
$ \hat A \to \hat U \hat A \hat U^\dagger $.  
Because $\hat  U$ is unitary, $\hat A$ and $\hat U \hat A \hat U^\dagger$ share the same 
eigenvalue spectrum.  It's as if $\hat U$ performs a rotation of $\hat A$ in U$(N)$.  

As shown in appendix \ref{ap:th1}, a shift of the parameters $b$ and $\phi_0$ by 
factors of $\frac{2 \pi}{N}$ is equivalent to a conjugation 
of both Hamiltonian and resulting Floquet propagator by factors of unitary operations 
$\hat X$, $\hat Z$,  called  clock and shift operators
 (defined in equations \ref{eqn:XY}) which are a set of generalized Pauli matrices.  
The set of unitary Floquet operators generated from 
$b, \phi_0 \in \{  \frac{2 \pi j}{N}: j \in {\mathbb Z}_N \}$ is equal to the set $S_{HW}$ generated
from the Floquet propagator with $b=0, \phi_0=0$ and 
defined via conjugation defined in appendix \ref{ap:twirlset} in equation \ref{eqn:S_HW}. 
In appendix \ref{ap:twirlset} we estimate the k-frame potentials of 
a uniform distribution on this set, finding that   
 ${\cal F}^{(k)} \gtrsim N^{2(k-1)}$.  For $k>1$ these are much larger than those of 
Haar-random distribution.  The large estimated values are consistent with the large k-frame 
potentials we measured for the sampler labelled $Db\phi$ in Table \ref{tab:frames}. 
Even though the conjugation effectively rotates the underlying propagator in about $N^2$ directions, 
 the resulting distribution of propagators 
consists only of unitary operators that have the same eigenvalues and this means
the  resulting distribution of unitaries differs significantly from a Haar-random distribution. 


The Floquet theorem (e.g., \citep{Rudner_2020}) states that a state vector that evolves via
a time-periodic Hamiltonian $H(t+ T) = H(t)$ with period $T$ is a sum of states in
the form $\ket{\psi_a(t)} = e^{i\lambda_a t} \ket{\phi_a(t)}$, where $ \lambda_a$ is the quasi-energy,
$\ket{\phi_a(t + T)} = \ket{\phi_a(t)}$, and the subscript $a$ labels the states. 
If the phase of perturbation $\tau_0$ in the perturbation $\hat h_1$ (equation \ref{eqn:hath1}) is varied, 
the quasi energies of the resulting Floquet propagator (computed
for a duration of the perturbation period $2 \pi$) do not vary.     
That means shifting  the perturbation phase $\tau_0$ in the Hamiltonian of our Floquet system is also a conjugation.   Describing the dependence of the Floquet propagator on the phase of
perturbation  as $\hat U(\tau_0)$, 
Floquet's theorem implies that 
\begin{align}
\text{ for any }  \alpha \in [0, 2\pi), \exists \hat V(\alpha)  \in \text{U}(N) 
\text{ such that } \nonumber \\
\hat U_T(\tau_0 = \alpha) =  \hat V(\alpha)   \hat U_T(\tau_0 = 0)  \hat V^\dagger(\alpha) .
\end{align} 
The set $\{ \hat U_T( \tau_0 = \alpha) : \alpha \in [0, 2 \pi)\}$ is a loop in the space U$(N)$.
The set of unitary operators $\hat U_T(\tau_0)$ with different 
values of the phases  all have the same eigenvalues (quasi-energies) 
due to the Floquet theorem.

The operations shifting $b$ or $\phi_0$ by $2 \pi/N$, or shifting $\tau_0$ by any phase,  
correspond to conjugation of the original unitary propagator by a unitary operation.
Hence samplers based on these three types of distributions 
would only generate unitaries that have the same eigenvalues as the original. 
That implies that samples constructed from such discrete distributions of 
$b, \phi_0$ or/and continuous distributions of $\tau_0$ cannot be Haar-randomly distributed as 
they cannot generate unitaries with different eigenvalues. 
This in part explains why the sampler $Db\phi$ (choosing $b, \phi_0$ from a discrete set) 
has a larger frame potential than the sampler $Cb\phi$ where $b, \phi_0$ are chosen 
from uniform distributions in $[0, 2\pi)$. 

In contrast, if we choose a distribution of perturbation parameters $\mu, \mu'$  
or use distributions of for $b, \phi_0$ with values not separated by factors of $2\pi/N$, 
then variations in the resulting unitary 
eigenvalue distributions is likely.   If we choose distributions of parameters that both rotate 
the unitaries and change their eigenvalues, the resulting sampler would be closer to uniform. 
This in part explains why the Floquet samplers generated from both distributions of 
parameters that affect the strength of perturbation ($\mu, \mu'$) and continuous distributions of angles 
have k-frame potentials closer to those of a Haar-random distribution. 

The samplers D$b\tau$C$\mu \mu'$ and C$_1b\tau\mu\mu'$ samplers are similar 
except the angles are chosen from multiples of $2\pi/N$ in the D$b\tau$C$\mu \mu'$ sampler 
and  from a uniform distribution 
in $[0,2\pi)$ for the C$_1b\tau\mu\mu'$ sampler.   
These samplers are both consistent with a Haar-random distribution. 
Sensitivity to $\mu \mu'$ which affects the propagator quasi-energies along with 
 conjugation due to shifts in angles $b, \tau_0$ is sufficient to give a near Haar uniform distribution independent of whether the angles are chosen from multiples of $2\pi/N$ 
 or from uniform distributions in $[0, 2 \pi)$. 
 
\subsection{Parameter sensitivity of Floquet evolution}
\label{sec:interaction}

In this section we heuristically illustrate with a Floquet-Magnus expansion 
why a variation of a few control parameters could allow the Floquet propagator to 
cover much of U$(N)$. 

We consider evolution via Schr\"odinger's equation with Hamiltonian operator 
$\hat h_0 + \hat h_1(t)$, a sum of a time-independent operator 
and a time-dependent perturbation.  
In the interaction picture,  state vectors are defined with a time-dependent
unitary transformation
\begin{align}
\ket{\psi_I(t)} = e^{\frac{i}{\hbar} \hat h_0 t} \ket{\psi(0)}.
\end{align}
The evolution of the perturbation operator in the interaction picture 
\begin{align}
\hat h_{1,I} (t) & = e^{\frac{i}{\hbar} \hat h_0 t} \hat h_1(t)  e^{-\frac{i}{\hbar} \hat h_0 t}.\\
&= \sum_{j=0}^\infty \left( \frac{it}{\hbar} \right)^j\frac{1}{j!}  [(\hat h_0)^{(j)}, \hat h_1 (t)] 
.\label{eqn:h1_I}
\end{align}
We have used the Campbell identity and
shorthand notation for repeated commutators 
(e.g., $ [(\hat h_0)^{(3)}, \hat h_1] = [\hat h_0,[ \hat h_0, [\hat h_0, \hat h_1]]]$ 
and $[(\hat h_0)^{(0)}, \hat h_1]  = \hat h_1$. 
With $\hat h_1(t)$ periodic with period $T$ and 
in the interaction picture, the Floquet propagator over period $T$ 
\begin{align}
\hat U_{T,I} = {\cal T} e^{-\frac{i}{\hbar} \int_0^T\hat h_{1,I}(t) dt } .
\end{align}

We set 
\begin{align}
\hat U_{T,I}  = e^{- \frac{iT}{\hbar}  \hat h_{F,I} }
\end{align}
where $\hat h_{F,I}$ is the time-independent effective or
average Hamiltonian.
The Floquet-Magnus expansion 
 \citep{Magnus_1954,Chu_1985,Blanes_2009,Brinkmann_2016,Kuwahara_2016}
 gives the effective Hamiltonian $\hat  h_{F,I}$ as 
 sum of a series of operators 
\begin{align}
 \frac{T}{\hbar} \hat h_{F,I} &= \sum_{n=1}^\infty \hat \Omega_n .  \label{eqn:hFI_def}
\end{align}
The first two terms in the expansion are 
\begin{align}
\hat \Omega_1 &=\frac{1}{\hbar} \int_0^{T} \hat h_{1,I}(t) dt\label{eqn:Om1_a} \\
\hat \Omega_2 &=\frac{1}{\hbar^2 2i} \int_0^{T}  dt_1 \int_0^{t_1} dt_2 \ 
 [\hat h_{1,I}(t_1) ,\hat h_{1,I}(t_2) ].  \label{eqn:Om2}
\end{align}
The $n$-th operator $\hat \Omega_n$ in the Magnus expansion depends upon 
 commutators that have $n$ factors of the perturbation operator $\hat h_{1,I}$. 
 
For the perturbation we used for the Harper model, (equation \ref{eqn:hath1}) 
\begin{align}
\hat h_1(\tau) &= -\mu \cos (\hat \phi - \phi_0 - \tau + \tau_0) -\mu'  \cos (\hat \phi - \phi_0 + \tau - \tau_0) \nonumber \\
& = -(\mu + \mu') \cos (\hat \phi -\phi_0) \cos (\tau - \tau_0)  \nonumber \\
& \ \  \ + (\mu' - \mu) \sin (\hat \phi -  \phi_0) \sin (\tau - \tau_0) .
          \label{eqn:h1_expand}
\end{align}
This  is a sum of parameters times operators and functions of time that are periodic 
\begin{align}
\hat h_1(\tau) = \sum_{j=1}^2 \mu_j \hat a_j  f_j(\tau)  \label{eqn:h1form}
\end{align}
with 
\begin{align}
\mu_1 &= -(\mu + \mu' ), \mu_2 = \mu' - \mu \nonumber \\
\hat a_1 &= \cos(\hat \phi - \phi_0), \hat a_2 = \sin (\hat \phi - \phi_0) \nonumber \\
f_1(\tau) &= \cos(\tau -\tau_0), f_2(\tau) = \sin(\tau -\tau_0). \label{eqn:our_exp}
\end{align}
 
We adopt a similar form to equation \ref{eqn:h1form} for a more general perturbation that has 
$m$ free parameters, each denoted $\mu_j$ with index $j \in \{1, \dots, m\}$ 
or equivalently  an $m$ length vector ${\boldsymbol \mu} \in {\mathbb R}^m$,  and with each term in
the expansion of $\hat h_1(t)$ separable w.r.t. to time; 
\begin{align}
\hat h_1(t) =  \sum_{j=1}^m \mu_j \hat a_j f_j(t) . \label{eqn:aexp}
\end{align}
Here the set $\{ \hat a_j : j\in  \{1, \dots, m\}$ is a set of Hermitian operators 
and the set of functions $\{ f_j(t)  : j\in  \{1, \dots, m\}$ are 
 periodic with period $T$. 
Inserting this form of time dependent perturbation into the relation for the operators in the Magnus expansion 
and using the Campbell identity (equation \ref{eqn:h1_I})
\begin{align}
\hat \Omega_1 & =\frac{1}{\hbar} \sum_{j=1}^m  \mu_j \int_0^{T} dt\  e^{ \frac{i}{\hbar} \hat h_0t} \hat a_j e^{- \frac{i}{\hbar} \hat h_0t}   f_j(t) \nonumber \\
& = \frac{1}{\hbar} \sum_{j=1}^m  \mu_j \sum_{k=0}^\infty [ (\hat h_0)^{(k)}, \hat a_j] 
\int_0^T dt \ \left(\frac{it}{\hbar} \right)^k \frac{f_j(t) }{k!} \\
\hat \Omega_2 & = \frac{1}{\hbar^2 2i}\sum_{j,j'=0}^m  \mu_j \mu_{j'} \sum_{k,k'=0}^\infty 
\Big[ [ (\hat h_0)^{(k)}, \hat a_j] ,[ (\hat h_0)^{(k')}, \hat a_{j'}] \Big] \nonumber \\
& \times \int_0^{T} dt_1 \int_0^{t_1} \left(\frac{it}{\hbar} \right)^{k+k'} dt_2 \frac{f_k(t_1) f_{k'}(t_2)}{k! (k')!}  .
\label{eqn:Omms}
\end{align}
Focusing on their dependence on the control parameters, we write the first two operators  in the Magnus expansion as 
\begin{align}
\hat \Omega_1 & =  \sum_{j=1}^m \mu_j \hat b_j \nonumber \\
\hat \Omega_2 & = \sum_{j=1}^m \sum_{k=j}^m \mu_j \mu_k \hat b_{jk}  \label{eqn:Omms2}
\end{align}
where  the operators $\hat b_j, \hat b_{jk}$ are given by factors and commutators in equations \ref{eqn:Omms}. 

The first operator in the Magnus expansion $\hat \Omega_1$ is a sum of $m$ operators, each constructed 
from commutators of $\hat h_0$ and another operator (one of $\hat a_j$) and 
each proportional to a different parameter ($\mu_j$). 
For the perturbation we use for our Floquet samplers, the two 
operators $\hat a_1, \hat a_2$ in the expansion of $\hat h_1(t)$ in equation \ref{eqn:our_exp} 
commute.  
However, their commutators with $\hat h_0$ would not necessarily commute.  
Consequently the individual operators in the sum of $\hat \Omega_1$, denoted $\mu_1 \hat b_1, \mu_2 \hat b_2$ in equation \ref{eqn:Omms2}, could be linearly independent. 
They  depend on the different control parameters so a distribution of 
control parameters would cover a 2 dimension space. If there are $m$ control parameters then 
by varying them, $\Omega_1$ could cover 
an $m$ dimensional space.  
Here we mean linear independence using the Hilbert-Schmidt inner product on the space of linear operators. 
If the individual operators $\hat a_j$ in the description of the 
perturbation don't commute with each other, the 
resulting effective Floquet Hamiltonian would probably be more sensitive to its parameters. 

The second operator $\hat \Omega_2$  in the Magnus expansion series 
is a sum of $ m(m+1)/2$ operators, each controlled by a 
product of two parameters in the parameter list $\boldsymbol \mu$.   
Because each term depends upon a product of two parameters in the parameter list, each term 
is controlled differently.   The operator $\hat b_{ij}$ is  constructed from 
products of pairs of perturbation operators $\hat a_i, \hat a_j$,  so the set of all of them
 $\{ \hat b_{ij}: \text{ for } i,j \in {\mathbb Z}_m \}$  could be a linearly independent set.  
As each is sensitive to a different polynomial in the control parameters $\boldsymbol \mu$, 
a distribution of control parameters would give an $m$-dimensional surface that 
could lie in an $m(m+1)/2$ dimensional space.   
 The operators $\hat b_{ij}$ depend upon commutators of $\hat a_i, \hat a_j$ and $\hat h_0$
whereas the operators $\hat b_j$ depend upon commutators of $\hat a_j$ and $\hat h_0$.
The set of operators $\{ \hat b_{ij}, \hat b_k: \text{ for }  i,j,k \in {\mathbb Z}_m \}$ is also likely to be linearly independent set.  

With two free parameters, 
 the surface in the $N^2$ dimensional complex space of linear operators
  explored by $\Omega_1 + \Omega_2$ 
 is two-dimensional.  However, it is not confined to a two dimensional subspace because it depends on the 5 
 operators  $\hat b_1, \hat b_2, \hat b_{11}, \hat b_{22}, \hat b_{12}$. It would ruffle and extend  
 out of a 2-dimensional subspace as $\mu_1, \mu_2$ are increased, as shown in Figure \ref{fig:illust2}. 
 
\begin{figure}[htbp]\centering 
\includegraphics[height=2truein]{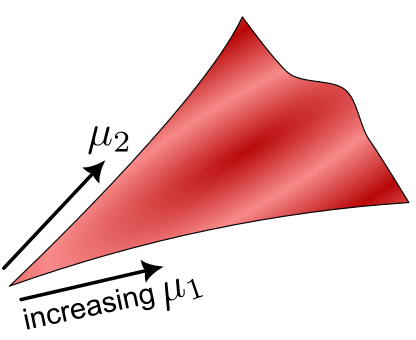}
\caption{ 
A linear operator is dependent upon two parameters $\mu_1, \mu_2$.  With a distribution 
of parameter values, the unitarity operators cover a 2-dimensional surface.  However, 
the surface can become ruffled with increasing $\mu_1, \mu_2$ and would not be confined to 
a two dimensional subspace within the $N^2$ dimensional complex space of linear operators.  
\label{fig:illust2}}
\end{figure}

 

Subsequent operators $\hat \Omega_k$ with $k>2$ would depend on 
a sum of ${\cal O}(m^k)$ separate operators, each controlled 
by a product of $k$ parameters in the vector of control parameters $\boldsymbol \mu$. 
The resulting set of  operators could span a space,  in the space of operators, by varying the parameters,  that lies in an ${\cal O}(m^2)$ dimensional space, even though the surface given by a distribution 
in $\boldsymbol \mu$ would be $m$-dimensional.  

A set of only two unitary operators can be a universal set in the sense that any unitary can be approximated by a product of a number of instances of each operator \citep{Sawicki_2017}. In other words, two unitary operators can be sufficient to generate the entire group U$(N)$ \citep{Sawicki_2017}. 
Multiple powers and  consecutive factors of only two operators 
 can produce linearly independent operators.   As each operator in the Magnus expansion contains
 different numbers of perturbation operators, the dimension of possible operators that could 
 be constructed from a sum dependent on different parameter values could be large. 

With a Magnus expansion truncated at order $K$, what value of $K$ would give approximately 
$N^2$ linearly independent operators?  
We set $m^K \sim N^2 $ to find 
\begin{align}
K \sim \frac{2 \ln N}{\ln m}.
\end{align}
The logarithmic dependence on the dimension of the quantum space $N$ 
 implies that only a few significant terms 
in the Magnus expansion are required to give an effective Floquet Hamiltonian $\hat h_F$ 
that can span the space of linear operators using its dependence on its control parameters. 
The factors of $1/\hbar$ in equation \ref{eqn:Omms} can cause each term in the effective Floquet Hamiltonian to be large enough that the propagator could experience strong phase wrapping in its quasi-energies as the perturbation parameters vary. 
As a consequence, a distribution of the control parameters would allow the resulting 
distribution of Floquet propagators to cover U$(N)$.   

The dimension of a space reached by varying control parameters
of a Floquet propagator differs from the dimension of a Krylov space. 
With a time independent operator $\hat V$, 
the evolution of this operator under evolution by Hamiltonian $\hat h_0$ is given equation \ref{eqn:h1_I}, 
\begin{align}
\hat V (t) & = e^{- \frac{i}{\hbar} \hat h_0 t} \hat V e^{ \frac{i}{\hbar} \hat h_0 t} \nonumber \\
& = 
\sum_{j=0}^\infty \left( \frac{it}{\hbar} \right)^j\frac{1}{j!}  [(\hat h_0)^{(j)}, \hat V ] .
\label{eqn:h1constant}
\end{align}
As a function of time $\hat V(t)$ explores 
a space, known as a Krylov space, that is spanned by the set of operators $[(\hat h_0)^{(j)}, \hat V]$ 
for all non-negative integer $j$  \citep{Krylov_1931,Rabinovici_2021,Nizami_2023}.
The part of the operator algebra that is absent from the Krylov
space belongs to the space subtended by the projectors created from eigenstates of $\hat h_0$. 
The first term in the Magnus expansion $\hat \Omega_1$ of our Floquet propagator contains $m$ operators, 
one generated from each perturbation term, 
but because of the averaging they are not individually 
dependent upon time dependent factors of nested commutators with $\hat h_0$. 
While both of these settings contain nested commutators, they differ. 
In the evolution of a single perturbation with time, (as in equation \ref{eqn:h1constant}) 
each commutator depends on a different power of time.  In contrast 
the Floquet system is dependent on perturbation control parameters rather than time, and 
the operators in the Magnus expansion are averaged. 
The complexity of the Floquet system depends on the independence of the averaged commutators
with respect to their control parameters.  

To ensure that the chaotic region covered phase space, we used strong perturbation strengths $\mu, \mu'$ 
in the construction of our samplers. The sizes of terms in the Magnus expansion 
would all be similar in size.   While this aids in ensuring that the resulting effective Floquet 
Hamiltonian is sensitive to every independent operator arising 
from commutators in each term of the expansion, the expansion itself would not converge. 
%
Whether or not the Magnus expansion converges may be related to 
the ergodicity of the system \citep{DAlessio_2016}. Even if the expansion diverges, if 
the Magnus expansion is truncated, it could be a good approximation \citep{Kuwahara_2016}. 
In this case the dimension spanned by the effective Floquet Hamiltonian would depend upon 
an optimal truncation index \citep{Abanin_2017}. 

Our study of the sensitivity of a Floquet system to system parameters is 
similar to the study of the Loschmidt echo operator which measures the divergence as a function of time 
of two propagators derived from slightly different Hamiltonians \citep{Brown_2018}.  
Divergence of two nearby operators as a function of time depends on 
$|\tr \left( e^{-i Ht} e^{i( H+\Delta H) t }\right)| $ 
where $\Delta H$ is a perturbation to $H$, whereas sensitivity to a control parameter $\mu$  would be 
described by 
$|\tr \left( e^{-iH(\mu)} e^{i H(\mu + \Delta  \mu)} \right)| $, where $\Delta \mu$ represents a change
in the control parameter $\mu$. 

The maximum Lyapunov exponent for the periodically perturbed pendulum depends on the strength 
 of the perturbation (e.g., \citep{Chirikov_1979,Shev_2024}).   With two terms, like  
 those of the classical perturbation $H_1$ in equation \ref{eqn:H1} that depend on $\mu, \mu'$, 
 the Lyapunov exponent  is sensitive to both parameters \citep{Shev_2024}.    
 In contrast the parameters $b, \phi_0$ only shift the time independent portion of the Hamiltonian. 
 The parameter $\tau_0$ only shifts the phase of the perturbation.  The Lyapunov exponent 
 in the classical system is not sensitive to these three angles. 
 
Consider a distribution of particle orbits  
that are chosen from a distribution of 
possible values for perturbation strength parameters $\mu, \mu'$ but are all initially at 
the same location in phase space at the same initial time. 
After a few Lyapunov times, each particle would be at a random location 
within its orbit.  If the ergodic regions for all $\mu, \mu'$ have the same volume 
in phase space, then at later times, 
the particle distribution would uniformly cover phase space.  
In the associated quantum system, we consider a transition probability 
between two localized states centered at two specific positions.  
If a classical orbit goes from one to the other location, then the transition probability of the propagator 
between the two localized states would be high.   
A relation between classical trajectories and quantum transition
probabilities implies that the probability of transitions between two localized states 
would strongly depend upon the perturbation strength parameters. 
A distribution of perturbation strength parameters would give a distribution of transition probabilities 
and equivalently a distribution of unitary propagators. 
With this idea in mind, we support the connection between a Lyapunov 
exponent in a classical system and the sensitivity of the quantum system to variations 
in its parameters, analogous to the proposals by \citet{Scott_2006,Roberts_2017}. 
The sensitivity of the quantum system to its control parameters makes it possible 
 generate a nearly uniform quantum sampler from distributions of its control parameters. 


\subsection{Sensitivity to dimension $N$} 
\label{sec:N}
 
\begin{figure}[htbp]\centering
\includegraphics[width=3.3truein,trim = 8 0 0 0,clip]{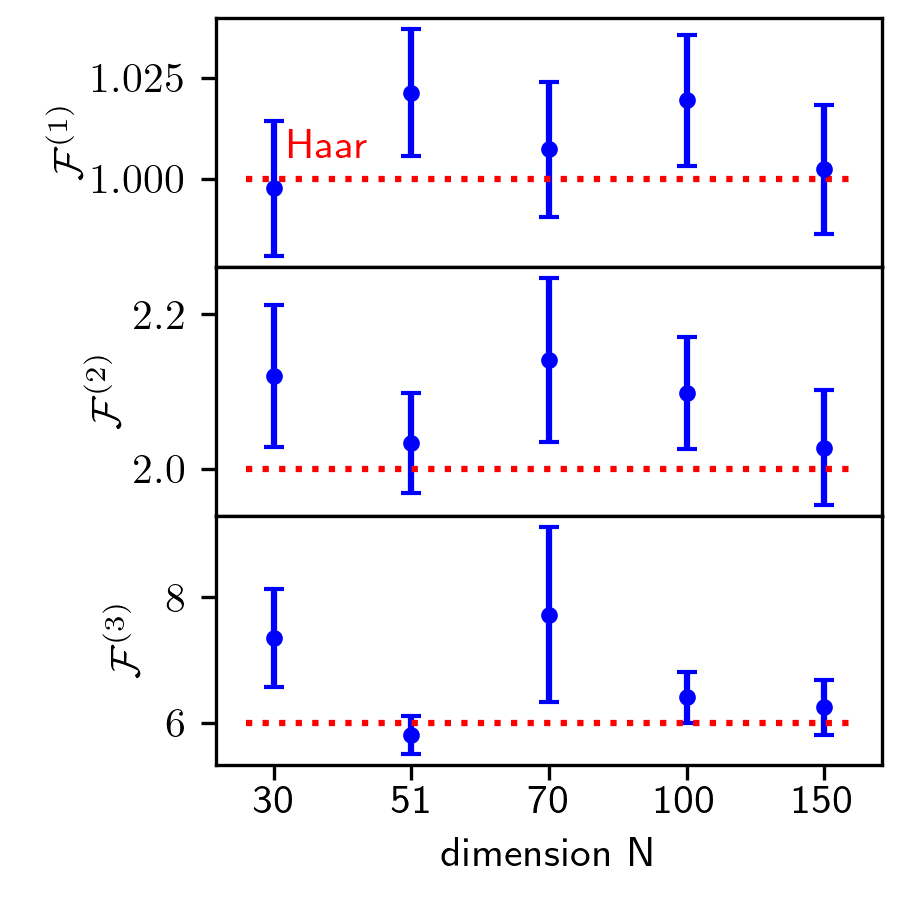}
\caption{Numerically estimated k-frame potentials of Floquet samplers 
generated by constructing propagators from distributions of parameters that describe them.  
This figure is similar to Figure \ref{fig:fp_fl} except we show frame potentials for 
the samplers C$_1b\tau\mu\mu' $, C$_2b\tau\mu\mu' $, C$_3b\tau\mu\mu'$, C$_4b\tau\mu\mu'$
and C$_5b\tau\mu\mu'$  which are the same except for their dimension $N$. 
The numerically estimated $k=1,2,3$ frame potentials are insensitive to $N$. 
\label{fig:fp_N}}
\end{figure} 
 
The Floquet samplers labelled C$_1b\tau\mu\mu' $, C$_2b\tau\mu\mu' $, 
C$_3b\tau\mu\mu'$, C$_4b\tau\mu\mu'$
and C$_5b\tau\mu\mu'$ in Table \ref{tab:fl_samplers} are similar except their dimensions are 
$N=$51, 30, 70, 100, and 150,  respectively.   The frame potentials for these 5 samplers are plotted in Figure \ref{fig:fp_N}. 
The similarity between their numerically estimated k-frame 
potentials suggests that the quality of this quantum sampler is not dependent on dimension $N$. 

The Floquet samplers shown in Figure \ref{fig:fp_N} are constructed with uniform distributions 
of two angles $b, \tau_0$ which approximately generate a set of unitaries via conjugation, as discussed
in section \ref{sec:conj}. 
In appendix \ref{ap:twirlset} we considered the set of 
unitaries constructed by conjugating an operator similar to a random matrix 
with the set of displacement operators.  We found that 
 the k-frame potential from a uniform distribution of unitaries in this set approximately 
 has upper and lower limits that are powers of dimension $N$. 
The largest dimension sampler we computed for Figure \ref{fig:fp_N} has $N=150$. 
If the k=1,2,3 frame potentials at dimensions $N<150$ 
are not strongly dependent on $N$,  as shown for a specific sampler in Figure \ref{fig:fp_N},  
could they increase at $N$ above 150?  
If the k-frame potentials depend on a power of $N$, then this is unlikely.  In that case 
a bad quantum Floquet sampler would have k-frame potentials that rapidly increase
with dimension $N$ whereas one that is close to a Haar random distribution would be insensitive to $N$. 

For the Harper model, the large $N$ limit is the semi-classical limit \citep{Quillen_2025}, so at large $N$ the quantum system should increasingly approximate the classical one.  In this limit,  the quantum Floquet sampler is related to a distribution of ergodic maps integrated over 1 Floquet period.  
Taking into account distributions of the two parameters that set the maximum Lyapunov exponent and the two angles that shift in phase space, we expect that the 1 period distribution of orbits would be uniformly distributed in phase space.  Consequently we would not expect the frame potentials of our best sampler to  increase past $N>150$. 
We suspect that our best Floquet samplers have frame potentials that are independent of dimension $N$, 
but we do not as yet have an analytical way to show that this is true. 

Our quantum sampler differs in many respects to random quantum circuits on arrays or lattices of qubit or qudit systems.   With $n$ qubits, the number of quantum states $N=2^n$. 
A local random quantum circuits acting on $n$ qubits in a linear chain composed of ${\cal O}(t^{10}n^2)$ many nearest neighbor two-qubit gates, each chosen from a Haar random distribution in $U(4)$, form an approximate unitary t-design \citep{Brandao_2016}.  In terms of $N$ this is equivalent to choosing ${\cal O} ( t^{10} (\log_2 N)^2)$ randomly chosen two-qubit gates.   The number of gates is proportional to the number of free parameters that must be chosen to describe an individual instance of the circuit.  This implies that an  approximation t-design based  on a local random quantum circuit requires choosing a number of free parameters that increases with 
the number of quantum states.   Low depth random local circuits can  
form approximate unitary designs \citep{Schuster_2025}, nevertheless the number of free parameters 
that must be chosen to specify a single sampled unitary operator depends on a function of $ \log N$. 
By introducing non-local gates that operate on two qubits that are not nearest neighbors on a particular lattice, 
the required number of gates to achieve a t-design obeys different scaling with $\log N$ 
but still increases with $N$ \citep{Harrow_2023}. 
To specify a particular unitary, 
the quantum random circuits approximate t-designs require choices of many more free parameters than our Floquet samplers.  
The number of free parameters for the quantum random circuit approximate t-designs increases with the dimension of the quantum
system, whereas we lack numerical evidence for such scaling with our best Floquet sampler. 

The method for constructing our Floquet samplers differs from local random quantum circuits.  
Our Floquet samplers have non-local interactions in time independent potential (in $\hat h_0$) 
and in the perturbations $\hat h_1(\tau)$.  To compute the Floquet propagator we require ${\cal O} (N)$ 
Trotterization steps.  If we viewed Trotterization steps as equivalent to gates, then perhaps our sampler would be described with $N$ dependent scaling.    It may be that non-local Floquet quantum samplers are more convenient or effective on certain types of quantum systems compared to random quantum circuits.  Future studies could work on developing quantitative predictions for the quality of Floquet based t-designs so that they can be better compared to t-designs constructed from quantum random circuits. 
 
\section{Summary and Discussion}

We have explored using an ergodic finite dimensional Floquet system to aid in constructing 
pseudorandom unitaries that approximate a distribution 
that is uniform with respect to the Haar measure. 
For the periodically perturbed Harper model,
we identified a regime,  
 that with a strong perturbation and with perturbation frequency a few times slower
than the libration frequency, that gives a wide ergodic 
region that covers phase space and lacks resonant islands. 
We compared the properties of 4 unitaries, two corresponding to 
periodic Floquet driven systems, one that is drawn from a Haar-random distribution 
and one that is based on a periodic system but is slowly drifting.  
The drifted and more strongly perturbed Floquet system exhibit many properties 
similar to that of the randomly chosen matrix.  We looked at 
inverse participation ratios as a way to measure anti-concentration, the tendency  
of Husimi distributions of the propagator eigenstates to cover phase space,
 the distribution of the eigenvalue spacings, 
and the distribution of transition probabilities between states in two different bases. 

We constructed distributions of unitaries, which we called samplers, 
by adopting distributions for unitary propagator control parameters.  
In this study we used the following recipe for constructing a quantum sampler:  
\begin{itemize}
\item We identified a chaotic system that has a broad ergodic region lacking resonant islands.  This can be done by studying a matching classical system or one could instead measure transition probabilities or inverse participation ratios to check that the operator has properties similar to one drawn from a random matrix ensemble.
\item
We chose a chaotic system that is sensitive to control parameters describing it. 
For our system, the Lyapunov exponents in the associated classical 
system depend on the perturbation strengths of periodic perturbations. 
\item 
Choosing parameters that are likely to both vary the eigenvalue distribution and rotate 
the unitary operator,   
we created a distribution of unitaries using distributions of these control parameters. 
\end{itemize}

We constructed samplers both from periodic (Floquet) systems and 
similar systems where parameters are allowed to slowly drift. 
For the Floquet samplers, 
distributions in 4 parameters are sufficient to generate a 
sampler that mimics a Haar-random distribution whereas for
the drifted system 3 parameters was sufficient. 
The best samplers  are made from distributions of angles and perturbation strengths. 
A distribution in angles tends to cause a distribution of rotations in U$(N)$ via conjugation 
(discussed in section \ref{sec:conj} and appendix \ref{ap:th1}) whereas the perturbation strengths 
are directly related to the Lyapunov exponent in the associated classical system. 
We characterized the quality of our samplers by computing k-frame potentials 
from numerically generated series of randomly generated unitary operators, and 
comparing the k-frame potentials to those predicted for a Haar-random distribution. 
Our best samplers have k-frame potentials similar to Haar-random unitaries and 
so are called frame potential $\varepsilon$-approximate k-designs. 

Remarkably, the number of randomly selected control parameters 
in our samplers is only a few and   
 is vastly exceeded by the $N^2$ dimension of the space of unitary 
operators U$(N)$.    (In most of our numerical experiments $N$ is between 30 and 70). 
In section \ref{sec:interaction} we discussed sensitivity of the effective Floquet 
Hamiltonian to perturbation strength parameters.     
In the context of the Floquet-Magnus
expansion in the interaction picture, the number of linearly independent terms in the expansion suggests 
that a time dependent perturbation composed of a few parameter dependent non-commuting operators
could cover U$(N)$.  We support the proposal
that sensitivity to initial conditions in a classical chaotic system manifests as sensitivity 
of an operator to its control parameters in the quantum setting (following \citep{Scott_2006,Roberts_2017}).
We have demonstrated that we can employ this sensitivity to generate quantum samplers.  


For our Floquet samplers, 
 a single sample is produced within a single perturbation period. For the drifted samplers  a  
 sample is produced after a few perturbation periods.    To ensure that the 
 chaotic region fills phase space, 
the perturbation period is related to the libration frequency in the potential well of the Harper 
model.  Nevertheless, the time required for sample 
generation with these samplers is probably short 
 compared to the time it takes to carry out other types 
of operations on a qudit quantum computer. 
Quantum sample generation via Floquet driving could be fast. 

Floquet driving could be used as an alternative technique, compared to random quantum circuits, for generating quantum samplers.  They may become useful for quantum simulation on qudit  
finite dimensional quantum computers  (e.g., \cite{Champion_2025}).  Using a similar approach to 
that explored here,  it may be possible to generate a quantum sampler that could be a component of a quantum computer with periodically driven spin systems or many body systems 
(e.g., \cite{Haake_1987,Russomanno_2015,Fava_2020}).
We only studied systems that are driven by a perturbation with a single time dependent Fourier component. 
Perturbations that contain 
additional Fourier components could aid in covering a larger volume 
in U$(N)$ and improve the quality of the derived samplers. 
Additional types of control perturbations, especially if their operators don't commute, upon a base chaotic system could give additional parameters which 
can aid in achieving a more uniform distribution of unitaries. 
Another way to achieve a more uniform distribution would be to consecutively multiply operators drawn 
separately from a randomly chosen Floquet driven system.  Instead of considering 
the set of operators $\hat U_T({\boldsymbol \mu})$ where we have 
a distribution of control parameters ${\boldsymbol \mu}$, we can consider the distribution 
 a series of $m$ unitaries, $\hat U_1 \hat U_2 \ldots \hat U_m$,  where 
 each operator is drawn separately from the same distribution of control parameters. 
 This type of sampler would likely be a better approximation 
 to a Haar-random distribution, following arguments based on rapid convergence of 
 Markov chain models on groups to stationary distributions 
 \citep{Oliveira_2009,Brandao_2016,Schuster_2025}.  
 Using multiples of the same operator drawn from a distribution of control parameters 
 would not give a uniform distribution due to Rain's theorem \citep{Rains_2003}, as quasi-energies 
 would become uniformly distributed in $[0,2\pi)$ and would not be distributed like those of the circular unitary ensemble.

Extensions to infinite dimensional systems may be difficult.  
While the superconducting transmon has a cosine potential, 
it is not confined to the torus. Even though the quantized Harper model can be used
to approximate the quantized pendulum \citep{Quillen_2025b}  in an infinite dimensional Hilbert space, 
a chaotic region generated by a periodic perturbation would necessarily be bounded in phase space and its 
size could be sensitive to its control parameters.   
Not all types of perturbations are necessarily suitable for generating quantum samplers. 
For example, Anderson localization might prevent some types of perturbations from giving effective samplers. 
We found that the number of linearly independent terms in a Magnus-Floquet expansion 
of a Floquet propagator could be large. However, because this expansion does not 
necessarily converge, we lack quantitative estimates for 
how well a distribution of control parameters causes the resulting distribution of Floquet propagator to 
cover the unitary group in dimension $N$. 
These topics could be studied in more detail in future work.  

\vskip 1.0 truein

Figures in this manuscript are generated with python notebooks available 
at \url{https://github.com/aquillen/Qsample}  \citep{Quillen_Qsample_github}.  

Acknowledgements: 
We thank Machiel Block, Arjun Krishnan, Chaitanya Murthy and Stephen Li for helpful discussions. 
We thank the University of Rochester, Research Presentation Awards for support
that allowed A. S. Miakhel to present a preliminary version of this work at the 12-th Rochester Conference on
Coherence and Quantum Science (CQS-12), June 2025.

\bibliographystyle{elsarticle-harv}
\bibliography{Qchaos}

\appendix 

\section{Operator basis coefficients} 
\label{ap:op_coefs}

In this section using a basis for operators, 
we examine the distributions of the operator coefficients in this basis of the 4 unitaries 
discussed in section \ref{sec:comp4} and listed in Table \ref{tab:props}. 

\begin{figure*}[htbp]\centering
\includegraphics[width=5.0truein]{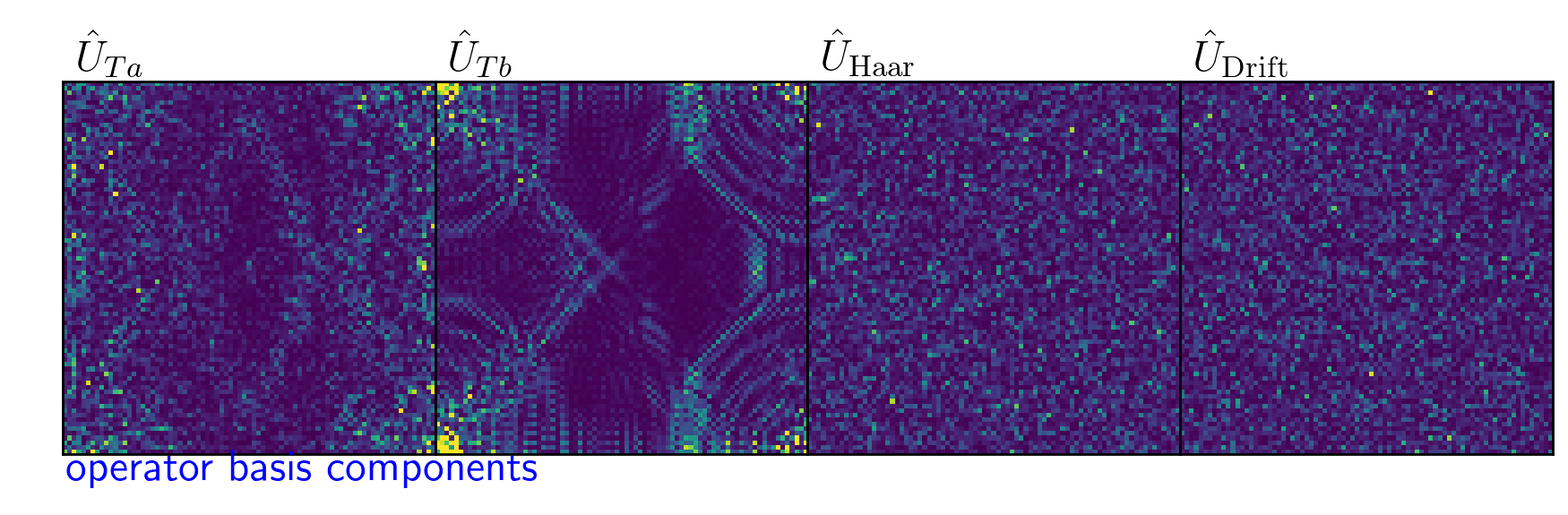}
\caption{Each panel shows all operator decomposition coefficients 
$|w_{jk}|^2$  (defined in equation \ref{eqn:wjk})  for the 4 
operators $\hat U_{Ta}$, $\hat U_{Tb}$, $\hat U_{Ha}$, $\hat U_\text{Drift} $ (with properties 
listed in Table \ref{tab:props}) computed 
in dimension $N=81$.  Each pixel in each panel shows
a specific $|w_{jk}|^2$ value with index $j$ increasing on the horizontal axis and $k$ increasing on the 
vertical axis. 
The hybrid (not fully ergodic) Floquet propagator $\hat U_{Tb}$ shows more structure 
than the Floquet propagator $\hat U_{Ta}$ which is fully ergodic.  The Haar-random unitary $\hat U_\text{Haar}$ and the drifted propagator $\hat U_\text{Drift}$ lack substructure. 
\label{fig:Wvals} }
\end{figure*}

\begin{figure}[htbp]\centering
\includegraphics[width=3.0truein]{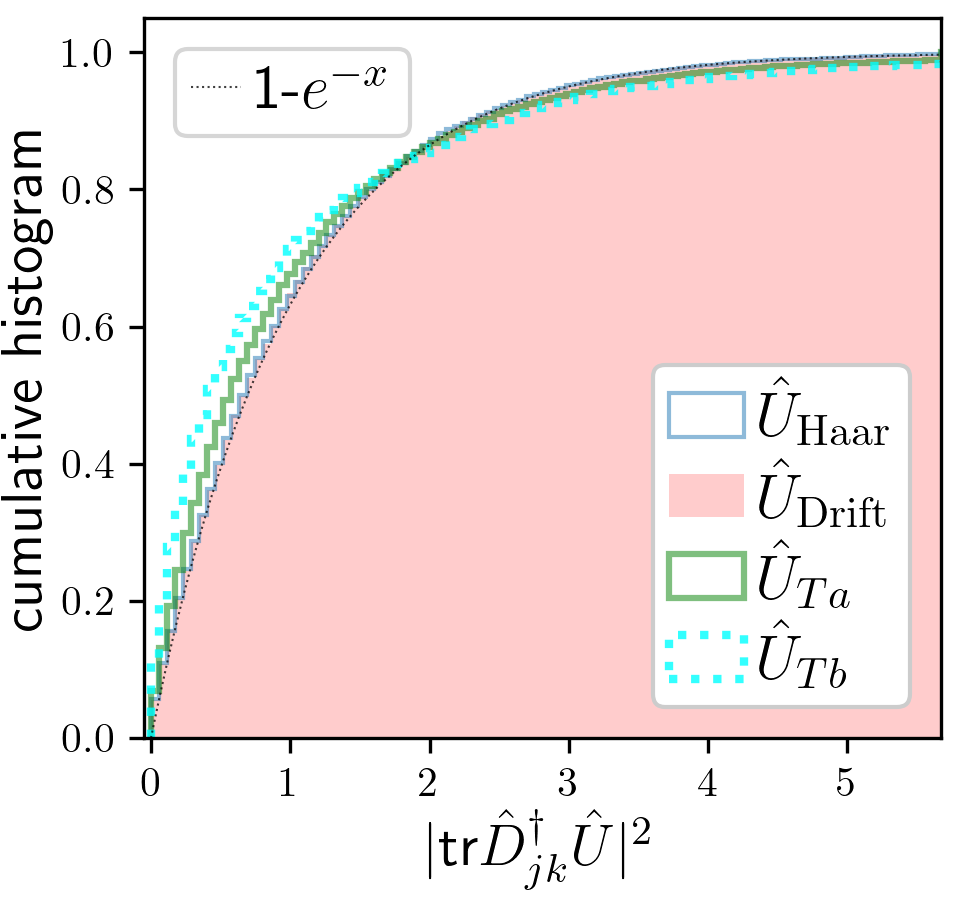}
\caption{Cumulative histograms of the sets $S_\text{op}(\hat U)$ (equation \ref{eqn:S_op}) of the square of the amplitude of the coefficients 
in the operator decomposition for each of the 4 unitary propagators listed in Table \ref{tab:props} with 
dimension $N=81$.
The coefficients have a cumulative distribution that resembles the exponential curve
which is plotted as a thin gray dotted line. 
The distributions for the Floquet operators $\hat U_{Ta}, \hat U_{Tb}$ differ 
from those of $\hat U_\text{Haar}$ and $\hat U_\text{Drift}$ which are closer to the exponential 
curve expected for a Porter-Thomas distribution.  
\label{fig:wcumu} }
\end{figure}


The unitary operators $\hat X, \hat Z$,  \citep{Schwinger_1960} 
are a set of generalized Pauli matrices called clock
and shift operators or Weyl-Heisenberg matrices
\begin{align}
\hat X &= \sum_{j=0}^{N-1} \ket{j+1} \bra{j}  \\
\hat Z & = \sum_{j=0}^{N-1} \omega^j \ket{j}\bra{j} \\
\omega & = e^\frac{2 \pi i}{N} \label{eqn:XY}
\end{align}
for an $N$-dimensional quantum system. 
Here the state $\ket{j+1}$ is computed with arithmetic modulo $N$. 
The clock and shift operators obey the commutation relationship
\begin{align}
\hat Z \hat X = \omega \hat X \hat Z.  \label{eqn:com_XZ}
\end{align}

Displacement operators are defined as 
\begin{align}
\hat D_{jk} \equiv \omega^{-\frac{jk}{2}}  \hat Z^j \hat X^k  \label{eqn:Djk}
\end{align}
for $j,k \in  {\mathbb Z}_N  $ 
and resemble those used for continuous systems of coherent states  \citep{Balazs_1989,Galetti_1996}.
Using the commutator of equation \ref{eqn:com_XZ}
\begin{align}
\hat D_{jk}  \hat D_{mn } 
& = \hat D_{j+m,k+n} \omega^{\frac{1}{2}(jn-km) }. \label{eqn:Dcom}
\end{align}
We gain a phase which is not important when taking a trace of a single operator
but could be important  when computing a sum within  the trace operator. 
The self-adjoint of the displacement operator 
\begin{align} 
\hat D_{ab}^\dagger  
& = \omega^{- \frac{ab}{2}}  \hat  Z^{-a} \hat  X^{-b}  = \hat D_{-a,-b} 
\end{align} 
where $\hat Z^{-a}$ is the inverse of $\hat Z^a$. 

Using the Hilbert-Schmidt inner product for operators, 
the displacement operators are 
orthogonal  \citep{Schwinger_1960} 
\begin{align}
\tr \left( \hat D_{ab}^\dagger \hat D_{jk} \right) & = \tr \left( D_{j-a,k-b} \omega^{\frac{1}{2}(jb-ka) }  \right)
\nonumber \\
& =  N \delta_{ja}  \delta_{kb} .
\end{align}
An orthonormal basis for linear operators, with respect to the trace norm, is the set of unitary operators 
\begin{align}
\bar G_\text{HW} \equiv
\left\{ \frac{1}{\sqrt{N} } \hat D_{jk}: \text{ for } j,k \in  {\mathbb Z}_N    \right\} 
\end{align}
\citep{Schwinger_1960}.  
The set $\bar G_\text{HW} $ is related to the Heisenberg-Weyl group $G_\text{HW}$ which is the discrete 
subgroup of U$(N)$ generated by the clock and shift operators $\hat X, \hat Z$.   
The set of displacement operators 
$\bar G_\text{HW} $ is equivalent to the quotient of the Heisenberg-Weyl group and 
the discrete group containing the phase $\omega=e^\frac{2 \pi i}{N}$ and its integer powers. 
 
A decomposition of an operator $\hat W$ with the operator
basis generated from the set of displacement operators  
\begin{align}
\hat W & = \frac{1}{\sqrt{N}} \sum_{jk} w_{jk} \hat D_{jk} 
. 
\end{align}
Equivalently the coefficients in the operator decomposition 
(based on the displacement operators) of an operator $\hat W$  are 
\begin{align}
w_{jk} = \frac{1}{\sqrt{N}} \tr (\hat D_{jk}^\dagger \hat W ). \label{eqn:wjk}
\end{align}

If $\hat W$ is  unitary, then 
$\hat W \hat W^\dagger = {\hat I}$, where ${\hat I} = \hat D_{00}$ is the identity operator. 
This implies that 
\begin{align}
\sum_{j,k=0}^{N-1} w_{jk} w_{jk}^* &= N. \label{eqn:wsum_unitary}
\end{align}
Consequently  a uniform probability distribution of unitary operators with respect to the Haar measure  can be described as a uniform distribution of normalized vectors in a complex vector space  of dimension $N^2$.   
 
Equation \ref{eqn:wsum_unitary} implies that 
the coefficients $w_{jk}$ in the operator sum decomposition 
of a Haar-uniform distribution of unitary operators also 
obey a Porter-Thomas probability distribution; 
\begin{align}
p( |w_{jk}|^2) \approx  N  e^{- N {|w_{jk}|^2} } \label{eqn:pwjk}
\end{align}
(taking into account normalization). 
Equivalently the probability distribution 
\begin{align}
p_\text{Haar} \left[  | \text{tr} \hat D_{jk}^\dagger \hat U|^2 \right] \approx e^{-  | \text{tr} \hat D_{jk}^\dagger \hat U|^2}   
\end{align}
for any $j,k$ for a Haar-random distribution of unitary matrices $\hat U$. 

There is a difference between the probability of a single coefficient $|w_{jk}|^2$ 
from a distribution of operators and the distribution of coefficients 
constructed from all coefficients of a single matrix operator. 
Consider a quantum state $\ket{\psi}$ chosen from a uniform distribution in $\mathbb{C}^N$ (and normalized) 
and the set of values 
\begin{align}
S_{\ket{\psi}} = \left\{ p_j = |\bra{j}\ket{\psi}|^2: \text{ for } j \in   {\mathbb Z}_N \right\} . 
\end{align}   
The set contains 
the probabilities of measurements in the different basis states $\ket{j}$.  
Except for the constraint due to normalization, each value of $p_j$ is nearly 
independent of all the others.  Hence if we choose a random state $\ket{\psi}$ 
and examine the set  $S_{\ket{\psi}}$,  a histogram of the values in $S_{\ket{\psi}}$   should resemble the 
Porter-Thomas distribution for almost all (in the sense of the Haar measure) 
operators in a set of randomly chosen states $\ket{\psi}$. 

Using the displacement operators to carry out an operator decomposition for an operator $\hat U$
we construct the set of probabilities 
\begin{align}
S_\text{op} (\hat U) &= \left\{ N |w_{jk}|^2 :  \text{ for } j,k \in  {\mathbb Z}_N   \right\} 
\label{eqn:S_op}
 \\
& =  \left\{ |\text{tr}  \hat D_{jk}^\dagger \hat U |^2 :    \text{ for } j,k \in  {\mathbb Z}_N \right\}. \nonumber 
\end{align}
The set $S_\text{op}(\hat U)$ are shown as images for the 4 unitaries in Table \ref{tab:props} 
in Figure \ref{fig:Wvals}.   For each panel, the index $j$ is increased along the horizontal axis
and $k$ is increased along the vertical axis, and  each pixel shows a single value for $N |w_{jk}|^2$. 
Figure \ref{fig:Wvals} illustrates that the coefficients in the operator basis are uniformly distributed 
for the Haar-random unitary $\hat U_\text{Haar}$ and the drifted propagator $\hat U_\text{Drift}$. 
The Floquet propagator $\hat U_{Tb}$ that contains both ergodic and non-ergodic regions in phase space
contains more substructure than the Floquet propagator $\hat U_{Ta}$ which is fully ergodic. 

Cumulative histograms for the set of probabilities $S_\text{op} (\hat U)$ for the 4 unitaries in Table \ref{tab:props}  are shown in Figure \ref{fig:wcumu}.  On Figure \ref{fig:wcumu} a dotted black 
line shows the exponential curve expected for a Porter-Thomas distribution.  
The similarity between the cumulative histogram of the $\hat U_\text{Haar}$ 
unitary and the Porter-Thomas distribution confirms our expectation that the set $S_\text{op}$, constructed 
from coefficients of a single unitary in an operator basis that was drawn from  a Haar uniform 
distribution is likely to resemble a Porter-Thomas distribution. 
The cumulative distribution of the set $S_\text{op}( \hat U_\text{Drift}) $ 
is indistinguishable from that of  $S_\text{op}( \hat U_\text{Haar}) $ suggesting that $U_\text{Drift}$
is similar to a random matrix.  There are small but noticeable deviations from 
Porter-Thomas distribution exhibited by the sets $S_\text{op}( \hat U_\text{Ta}) $  and $S_\text{op}( \hat U_\text{Tb}) $.  

\begin{table}[htbp] \centering
\caption{Moments of  operator basis coefficients\footnote{Moments for the operator basis coefficients are defined in equations \ref{eqn:wjk}
and \ref{eqn:mms_Op}. } \label{tab:mom_ops}}
\begin{ruledtabular}
\begin{tabular}{lllll}
 Unitary &  - & $ m_\text{Op}^{(2)}$ &  $m_\text{Op}^{(3)} $ \\
 \colrule
$\hat U_{Ta}$ && 2.50 $\pm$ 0.11 & 11.58$\pm$ 1.35 \\
$\hat U_{Tb}$ && 3.82 $\pm$ 0.30 & 37.10 $\pm$ 6.34 \\ 
$\hat U_\text{Drift}$ && 1.99 $\pm$  0.05 & 5.88  $\pm$ 0.29 \\
$\hat U_\text{Haar}$ && 2.00 $\pm$ 0.05 & 5.93 $\pm$ 0.28 \\
\colrule
Porter-Thomas\footnote{If the operator coefficients 
obeyed a exponential (Porter-Thomas) distribution, they would have these values. }   && 2  & 6 \\
\end{tabular}
\end{ruledtabular}
\end{table} 

As we did for the transition probabilities (see equation \ref{eqn:mms} in section \ref{sec:trans}), 
for each operator  we measure the moments of the distribution of 
operator basis coefficients 
\begin{align} 
m_\text{Op}^{(2)} &= \frac{1}{N^2} \sum_{j,k=0}^{N-1} \left(N|w_{jk}|^2\right)^2 \nonumber \\
m_\text{Op}^{(3)} &= \frac{1}{N^2} \sum_{j,k=0}^{N-1} \left(N|w_{jk}|^2\right)^3  ,
         \label{eqn:mms_Op}
\end{align}
where the coefficients are computed via equation \ref{eqn:wjk} from a unitary operator. 
We don't compute the first moment because it is 1. 
The moments for the 4 unitaries shown in Figures \ref{fig:Wvals} and \ref{fig:wcumu} are listed in Table \ref{tab:mom_ops}.
As was true for the transition probabilities, 
the moments of the propagator $\hat U_\text{Drift}$ for the drifting system
resemble those of a Haar-random matrix and an exponential distribution, 
and those of ergodic Floquet propagator $\hat U_{Ta}$ 
are larger but not as large as that of the hybrid Floquet propagator $\hat U_{Tb}$. 

\section{The relation between discrete shifts of parameter angles $b$ and $\phi_0$ and  unitary conjugation} 
\label{ap:th1}

We apply and extend Theorem A.3 by \citet{Quillen_2025b}.  

\begin{theorem} \label{th:th1}
We consider the Hamiltonian operator $\hat h$ of equations \ref{eqn:hath}, 
\ref{eqn:hath0} and \ref{eqn:hath1} which is a function of 
parameters $a, b, \epsilon, \mu, \mu', \phi_0, \tau_0$, and time $ \tau$. 
We fix $a, \epsilon, \mu, \mu', \tau_0$ and write $\hat h(b,\phi_0)$ because we will consider 
different values of parameters $b$ and $\phi_0$. 

Shifting the angles $b$ and $\phi_0$ by integer factors of $\frac{2 \pi}{N}$ is equivalent to 
the conjugation of $\hat h$ by factors of clock and shift operators $\hat X $ and $\hat Z$ (defined 
in equations \ref{eqn:XY}); 
\begin{align}
\hat h \Big(b = b_0+\frac{2 \pi j }{N}, &  \phi_0 =\alpha_0 +  \frac{2 \pi k }{N} \Big) =   \label{eqn:hshift} \\
&  \hat Z^{j}  \hat X^{k} \hat h(b= b_0, \phi_0=\alpha_0) \hat X^{-k}  \hat Z^{-j} .\nonumber
\end{align}

\end{theorem}



\begin{proof} 

The unperturbed term in the Hamiltonian operator (in equation \ref{eqn:hath0})
\begin{align}
\hat h_0(b,\phi_0) & = a(1-\cos (\hat p -b))  - \epsilon \cos (\hat \phi - \phi_0) \\
&=  a -  \frac{a}{2} \hat X e^{ib} -  \frac{a}{2}  \hat X^\dagger e^{-ib} \nonumber \\
& \ \ \ - \frac{\epsilon}{2} \hat Z e^{-i\phi_0} - \frac{\epsilon}{2} \hat Z^\dagger e^{i\phi_0}
.
\end{align}
On the right we have written the operator in terms of the clock and shift operators 
(following appendix A by \citep{Quillen_2025b}).  
Using the commutation relation for the clock $\hat Z$ and shift $\hat X$ operators (equation \ref{eqn:com_XZ})
\begin{align}
\hat Z^{-k} \hat h_0\left(b=b_0\frac{2 \pi k}{N}, \phi_0\right) \hat Z^k =  \hat h_0(b=b_0,  \phi_0) .
\label{eqn:Zh0}
\end{align}
Likewise we compute 
\begin{align}
\hat X^{-k} \hat h_0\left(b,  \phi_0 =\alpha_0+ \frac{2 \pi k}{N}\right) \hat X^k =  \hat h_0(b,  \phi_0=\alpha_0) .
\label{eqn:Xh0}
\end{align}

The perturbation portion of the Hamiltonian  (equation \ref{eqn:hath1})
\begin{align}
\hat h_1(&\phi_0) = -\mu \cos (\hat \phi  - \phi_0 + \tau - \tau_0)  \nonumber \\
&\qquad \qquad  \ \ \ -\mu' \cos (\hat \phi - \phi_0 - \tau + \tau_0)  \nonumber \\
& = -(\mu + \mu')  \cos( \hat \phi - \phi_0)\cos( \tau - \tau_0)  \nonumber \\
& \ \ \ + (\mu'-\mu) \sin( \hat \phi - \phi_0)\sin( \tau - \tau_0) \nonumber \\
& = - \frac{\mu + \mu'}{2}  ( \hat Z e^{-i\phi_0}  + \hat Z^\dagger e^{i\phi_0} ) \cos( \tau - \tau_0)\nonumber \\
 & \ \ \     + \frac{\mu'-\mu}{2i} (\hat Ze^{-i\phi_0} - \hat Z^\dagger e^{i\phi_0} ) \sin( \tau - \tau_0).
\end{align}
Consequently 
\begin{align}
\hat X^{-k} \hat h_1\left(\phi_0 =\alpha_0+ \frac{2 \pi k}{N}\right) \hat X^k 
=  \hat h_1(\phi_0 = \alpha_0) .
\label{eqn:Xh1}
\end{align} 

Putting together equations \ref{eqn:Xh0} and \ref{eqn:Xh1}  and with 
the full time dependent Hamiltonian $\hat h = \hat h_0 + \hat h_1$
we find that 
\begin{align}
\hat Z^{-j}\hat h\left(b=b_0 +\frac{2 \pi j }{N},\phi_0 \right)\hat Z^{j} &= \hat h(b=b_0,\phi_0) \\
\hat X^{-k}\hat h\left(b,\phi_0= \alpha_0 + \frac{2 \pi k}{N} \right)\hat X^{k} &= \hat h(b,\phi_0=\alpha_0) 
\end{align}
for  $j,k \in  {\mathbb Z}_N $. 
Conjugating the Hamiltonian operator by a factor of $\hat Z$, even when time dependent, is 
equivalent to shifting $b$ by a factor of $2 \pi /N$ and conjugating the Hamiltonian operator by a factor of $\hat X$ is equivalent to shifting $\phi_0$ by a factor of $2 \pi /N$. 
The Hamiltonian operator is a first order polynomial in the operators $\hat X, \hat X^\dagger, \hat Z$ and $\hat Z^\dagger$,  
so each term in the sum either commutes with $\hat X$ or commutes with $\hat Z$. 
Consequently 
\begin{align}
\hat X^{-k}  \hat Z^{-j} \hat h \Big(b= \frac{2 \pi j }{N},   \phi_0=& \frac{2 \pi k }{N} \Big)\hat Z^{j}  \hat X^{k} 
    \qquad \nonumber \\
  &  =\hat h(b=0, \phi_0=0) .\label{eqn:sym}
\end{align}
%
We rearrange equation \ref{eqn:sym} 
\begin{align}
\hat h \Big(b = b_0+ &\frac{2 \pi j }{N},   \phi_0 =\alpha_0 +  \frac{2 \pi k }{N} \Big)  \qquad  \ \ \ \nonumber \\
& = \hat Z^{j}  \hat X^{k} \hat h(b= b_0, \phi_0=\alpha_0) \hat X^{-k}  \hat Z^{-j} . 
\end{align}
which is equivalent to equation \ref{eqn:hshift}. 

\end{proof}

\begin{corollary}

Let 
\begin{align}
\hat U_T(b,\phi_0) = {\cal T} e^{-\frac{i}{\hbar} \int_0^T \hat h(b,\phi_0, \tau) d\tau} \label{eqn:UTT}
\end{align}
be the propagator constructed from the time dependent Hamiltonian $\hat h$.
We have written the parentheses of $\hat h$ and $\hat U_T$ to only show the variables that we would like to set.  

A shift of $b$ and $\phi_0$ by multiples of $\frac{2 \pi }{N}$ is equivalent 
to the conjugation of $\hat U_T$ by multiples of clock and shift operators; 
\begin{align}
\hat U_T \Big(b= b_0 +  \frac{2 \pi j }{N} , & \phi_0=\alpha_0 +  \frac{2 \pi k }{N} \Big) \label{eqn:coro} \\
&= \hat Z^{j}  \hat X^{k} \hat U_T(b=b_0 , \phi_0=\alpha_0) \hat X^{-k}  \hat Z^{-j} .\nonumber
\end{align}



\end{corollary}

\begin{proof}
For any unitary matrix $V$ and a Hermitian matrix $A$, 
the exponential $e^{VAV^\dagger} = Ve^{A}V^\dagger$ (a form of the Campbell identity). 
Each portion of the exponent in equation \ref{eqn:UTT}  obeys the same symmetry 
as in equation \ref{eqn:sym} hence 
\begin{align}
X^{-k}  \hat Z^{-j}\hat U_T\Big(b= &b_0+  \frac{2 \pi j }{N},  \phi_0=\alpha_0 +  \frac{2 \pi k }{N}  \Big)\hat Z^{j} X^{k} \ \ \ \ 
\nonumber   \\
& = \hat U_T(b =b_0,\phi_0 = \alpha_0 ) .\label{eqn:sym2}
\end{align}
We rearrange equation \ref{eqn:sym2} 
\begin{align}
\hat U_T \Big(b=& b_0 +  \frac{2 \pi j }{N},  \phi_0=\alpha_0 +  \frac{2 \pi k }{N} \Big) \ \ \nonumber \\
&= \hat Z^{j}  \hat X^{k} \hat U_T(b=b_0, \phi_0=\alpha_0) \hat X^{-k}  \hat Z^{-j} .
\end{align}
This is the desired result of equation \ref{eqn:coro}. 

The propagator could have 
 drifting system values of $a, \epsilon, \mu, \mu'$, or $\tau_0$. 
  If  $b$ or $\phi_0$  are drifting then the initial values of $b$ and $\phi_0$ in the propagator 
are shifted by conjugation with the clock and shift  operators. 

\end{proof}

\section{An estimate for the size of frame potentials of a conjugated set using the Heisenberg-Weyl group}
\label{ap:twirlset}


The operation known as a {\it twirl}, used in quantum information theory  
and in characterization of errors in quantum computing,   
 involves an average over elements that are conjugated 
by operators drawn from a subgroup of U$(N)$ such as the generalized Pauli group
\citep{Emerson_2005b,Gross_2007,Dankert_2009}. 
We similarly create a set of unitaries using conjugation 
over the set of displacement operators.  
From a particular unitary operator $\hat W$, 
we create the set of unitaries 
\begin{align}
S_{HW}(\hat W) = \{  \hat D_{mn}\hat  W \hat  D_{mn}^\dagger: \text{ for } m,n \in {\mathbb Z}_N \}  
\label{eqn:S_HW}
\end{align}
with displacement operators $\hat D_{mn}$ defined in equation \ref{eqn:Djk}. 
As shown in appendix \ref{ap:th1} and discussed in section \ref{sec:conj}, 
if parameter angles $b$ and $\phi_0$ are drawn from a uniform distribution of values 
$\{ \frac{2 \pi k}{N}, \text{for } k\in {\mathbb Z}_N \}$, then the resulting distribution of Floquet propagators   
is a uniform distribution of the unitaries in the set $S_{HW}(\hat U_T)$ where $\hat U_T$ is the Floquet propagator with $b=\phi_0 = 0$. 
In this section we estimate the k-frame potential for a uniform distribution of unitaries in this set.  

A decomposition of an operator $\hat W$ with the operator
basis generated from the set of displacement operators  
\begin{align}
\hat W & = \frac{1}{\sqrt{N}} \sum_{jk} w_{jk} \hat D_{jk} 
\end{align}
giving coefficients $w_{jk}$. 
We compute the conjugation of $\hat W$ by a displacement operator $\hat D_{mn}$;  
\begin{align}
 \hat D_{mn} \hat W \hat D_{mn}^\dagger & = 
  \sum_{jk} \hat  D_{mn} \hat  D_{jk} \frac{w_{jk}}{\sqrt{N}}  \hat D_{mn}^\dagger   \nonumber \\
& = \sum_{jk} \omega^{mk -nj}  \frac{w_{jk} }{\sqrt{N}}    \hat D_{jk}  
. \label{eqn:DWD}
\end{align}

The k-frame potential (see definition \ref{def:Framep}) of a uniform distribution of unitaries in the set $S_{HW}$ (defined in equation \ref{eqn:S_HW}) 
\begin{align}
{\cal F}^{(k)}_{S_{HW}} &= \frac{1}{N^4} \sum_{abcd} 
 | \tr \hat D_{ab} \hat W\hat D_{ab}^\dagger (\hat D_{cd} \hat W \hat D_{cd}^\dagger)^\dagger |^{2k} \nonumber \\
& =  \frac{1}{N^2} \sum_{ab}  \left| \sum_{jn} w_{jn} w_{jn}^* \omega^{an - bj} \right|^{2k}
\end{align} 
where we have used equation \ref{eqn:DWD}. 
This is a sum of powers of the Fourier components of the amplitudes of the $\hat W$ matrix. 
The operator sum decomposition with displacement operators is directly related to 
the discrete Fourier transform in the operator basis \citep{Tyson_2003}. 

With shorthand for positive real numbers  $p_{jk}$ 
\begin{align} p_{jk}  = w_{jk} w_{jk}^*,
\end{align}
the k-frame potentials 
\begin{align}
{\cal F}^{(k)}_{S_{HW}} &=
  \frac{1}{N^2} \sum_{ab}  \left| \sum_{jn} p_{jn} \omega^{an - bj} \right|^{2k}\nonumber \\
  & =  \frac{1}{N^2}  \sum_{ab}\left( \sum_{jn,j'n'} p_{jn} p_{j'n'} \omega^{ a(n-n')-b(j-j')}  \right)^k.
\end{align}

The case $a=b=0$ gives a lower limit 
\begin{align}
{\cal F}^{(k)}_{S_{HW}} &\ge \frac{1}{N^2} \left( \sum_{jn,j'n'} p_{jn} p_{j'n'} \right)^k  .  \label{eqn:ge}
\end{align}
Using $|\omega^{ a(n-n')-b(j-j')} | \le 1$ the upper limit 
 \begin{align}
 {\cal F}^{(k)}_{S_{HW}} & \le \left( \sum_{jn,j'n'} p_{jn} p_{j'n'} \right)^k . \label{eqn:le}
 \end{align}
 
We assume that $\hat W$ is similar to a random matrix in the sense 
that its operator basis coefficients are well described by a exponential distribution 
as we found in section \ref{ap:op_coefs}.
We approximate each $p_{jn}$ as independent of the others 
and use the fact that $\langle p_{jn}  \rangle  \sim \frac{1}{N}$  
(following from equation  \ref{eqn:pwjk}) for an operator that 
is similar to a random matrix. 
Equations \ref{eqn:ge} and \ref{eqn:le} then give  
\begin{align}
N^{2(k-1)} \lesssim {\cal F}^{(k)}_{S_{HW}} \lesssim N^{2k}.  \label{eqn:twirl_estimate}
\end{align}
The approximate inequality holds for the conjugation set $S_{HW}(\hat W)$ 
for $\hat W$ resembling a random matrix in the basis used to construct the displacement 
operators. 
The strong dependence on $N$ in equation \ref{eqn:twirl_estimate}
implies that we expect 
\begin{align}
{\cal F}^{(k)}_{S_{HW}} \gg {\cal F}_{\mu_\text{Haar} }^{(k)} = k! \ \ \ \text{ for } k>1.
\end{align}

\end{document}